\documentclass[12pt,letterpaper]{article}
\usepackage[T1]{fontenc}
\usepackage{lmodern}
\usepackage[utf8]{inputenc}
\usepackage[italian,USenglish]{babel}
\usepackage{geometry}
\usepackage{booktabs}
\usepackage{amsmath}
\usepackage{amsthm} 
\usepackage{amsfonts}
\usepackage{dsfont}
\usepackage{cancel}
\usepackage{ amssymb }
\usepackage{bm}
\usepackage{enumitem}
\usepackage{mathtools}
\usepackage{changepage}
\usepackage{amsthm}
\usepackage{verbatim}
\usepackage{amsmath,amssymb}
\usepackage{graphicx}
\usepackage{emptypage}
\usepackage{newlfont}
\usepackage[all,cmtip]{xy}
\usepackage[bottom]{footmisc}
\usepackage[titletoc,title]{appendix}
\usepackage{chngcntr}
\usepackage{apptools}
\usepackage{listings}
\usepackage{color}
\usepackage{sgame, tikz} 
\usepackage{kpfonts}    
\usepackage{natbib}
\AtAppendix{\counterwithin{lem}{section}}
\usepackage{caption}
\usepackage{subcaption}
\usepackage{float}
\usepackage{thmtools,thm-restate}
\usepackage{epigraph}
\usepackage{xr-hyper}
\usepackage[nodayofweek,level]{datetime}
\usepackage[colorlinks=true,linkcolor=blue, allcolors=blue]{hyperref}
\usepackage{sgamevar}

\theoremstyle{plain} 
 
\newtheorem{prop}{Proposition}
 
\newtheorem{theorem}{Theorem}
\newtheorem{lemma}{Lemma}
\theoremstyle{definition} 
\newtheorem{ex}{Example}

\newtheorem{rmk}{Remark}

\theoremstyle{remark}

\DeclareMathOperator{\E}{\mathds{E}}
\renewcommand{\P}{\mathds{P}}

\newcommand{\R}{\mathds{R}}

\newcommand{\supp}{\text{supp}}
\renewcommand{\1}{\mathds{1}}

\newcommand\eqid{\stackrel{d}{=}}

\renewcommand*\d{\mathop{}\!\mathrm{d}}

\newcommand{\argmax}{\operatornamewithlimits{argmax}}

\usepackage{setspace}
\theoremstyle{definition}

\reversemarginpar
\usepackage[nameinlink]{cleveref}
\crefname{manualasm}{assumption}{assumptions}
\crefalias{prop}{proposition}
\crefname{claim}{claim}{claims}
\crefname{ex}{example}{examples}
\crefname{defn}{definition}{definitions}
\crefname{rmk}{remark}{remarks}

\title{Dynamic Threats to Credible Auctions\thanks{We thank Gregorio Curello,
Matteo Escudé, Alkis Georgiadis-Harris, Mathijs Janssen, Shengwu Li, Ludvig Sinander, as well as several seminar and conference audiences at Northwestern, the University of Chicago, the University of Rochester, Bocconi University, EC'25, the Capri in Theory Workshop, the SITE Conference, and the Boston Theory Mini-Conference for helpful comments and discussions.}}
\author{Martino Banchio\thanks{Universit\`a Bocconi, IGIER, and Google Research. \href{mailto:martino.banchio@unibocconi.it}{martino.banchio@unibocconi.it}} \and Andrzej Skrzypacz\thanks{Stanford University, Graduate School of Business. \href{mailto:skrz@stanford.edu}{skrz@stanford.edu}} \and Frank Yang\thanks{Department of Economics, Harvard University. \href{mailto:fyang@fas.harvard.edu}{fyang@fas.harvard.edu}}}
\date{\today}
\begin{document}

\maketitle

\begin{abstract}
A seller wants to sell a good to a set of bidders using a credible mechanism. We show that when the seller has private information about her cost, it is impossible for a static mechanism to achieve the optimal revenue. In particular, even the optimal first-price auction is not credible. We show that the English auction can credibly implement the optimal mechanism, unlike the optimal Dutch auction. For symmetric mechanisms in which only winners pay, we also characterize all the static auctions that are credible: They are first-price auctions that depend only on the seller's cost ex post via a secret reserve, and may profitably pool bidders via a bid restriction. Our impossibility result highlights the role of public institutions and helps explain the use of dynamic mechanisms in informal auctions.\\


\end{abstract}
\setcounter{page}{1}
\newpage

\section{Introduction}

In this paper, we revisit the question of credible auction design and show that when the seller has private information about her cost, it is not possible to implement the optimal (i.e., profit-maximizing) mechanism using a static auction. As we show, even the sealed-bid first-price auction is not a credible mechanism. We show that optimality requires a dynamic mechanism and that the ascending auction can be used to credibly implement the optimal mechanism.

We are motivated to study credible auction design because many informal auctions are decentralized, following bilateral communications between the seller and potential buyers. 
In many informal sales, bidders are contacted with opportunities to submit bids for a good and do not have a chance to inspect the bids of others or what communication the seller had with other bidders.
For example, in bidding for houses, it is often hard for bidders to monitor or audit brokers running auctions regarding bids submitted by others or the communication the brokers have with others. Beyond informal auctions, our study also sheds light on many online auctions where sellers also communicate with bidders through private, bilateral communications.   

Our results yield two important takeaways.
First, we explain why many informal auction/sale processes are dynamic rather than static: It is impossible for sellers to credibly commit to optimal static mechanisms, and hence under credibility, dynamic mechanisms may outperform static ones. Similarly, this may explain why, in select online markets, sellers implement dynamic mechanisms (for example, eBay uses an auction format similar to the ascending auction). Second, we highlight the role of institutions, such as centralized clearinghouses or public announcements, in ensuring optimal market operations. Simple \emph{public} announcements about offer timing or reserve prices provide the seller with sufficient credibility to restore the optimality of static auctions.

Our notion of credibility follows the approach in \citet{akbarpour2020credible}. We say that an auction is credible if the only profitable deviations a seller could have are such that they would be observable by some bidder. For example, if a seller attempts to run a second-price auction, it would not be credible because after observing all the bids, the seller could charge the winner a price between the highest and second-highest bids (for example, via a practice known as soft reserve prices), and no individual bidder would be able to detect that deviation. 

The key new element of our model is that we allow the seller to have private information about her cost and not be able to commit to public messages. For example, the seller of a house may have her own private value of owning the house, a freight broker may have private information about the lowest price that she is willing to accept, and an online seller may have a private outside option of selling the item offline.

Our first result is that in this environment the optimal sealed-bid first-price auction is not credible, and in fact has a profitable safe deviation for the seller with probability $1$ (\Cref{thm:fpa}). It involves a dynamic deviation by the informed seller. The basic intuition behind why the optimal first-price auction is no longer credible can be understood through the following example (we will explain later why with probability $1$ the seller always wants to deviate):

\begin{ex}\label{ex:fpa}
Suppose that there are two bidders. Each bidder $i$ has an independent private value $\theta_i \in \{1, 2\}$ (with equal probabilities). The seller has an independent private cost $\theta_0 \in \{0, 0.7\}$. The optimal mechanism can be implemented via a first-price auction as follows. If the seller has cost $0$, then she tells each bidder the Myersonian reserve $r(0) = 1$, and solicits a bid from $\{1, \frac{5}{3}\}$; if the seller has cost $0.7$, then she tells each bidder $i$ the Myersonian reserve $r(0.7) = 2$, and solicits a bid from $\{0, 2\}$. A bidder of a higher type always selects the higher bid of the two. However, we claim that this mechanism cannot be credible. The deviation works as follows. Consider the seller of cost $0$. Let the seller follow the mechanism for bidder $1$ and tell bidder $1$ that the reserve is $r(0) = 1$. Suppose that the seller gets a high bid $\frac{5}{3}$ from bidder $1$. Now, if the seller continues following the mechanism, then regardless of the bid from bidder $2$, the seller's payoff will be $\frac{5}{3}$. However, if the seller deviates by pretending to have cost $\hat{\theta}_0 = 0.7$ to bidder $2$, then the seller gets a payoff 
\[\frac{1}{2} \times \frac{5}{3} + \frac{1}{2} \times 2 > \frac{5}{3}\,,\]
since with probability $\frac{1}{2}$, bidder $2$ bids $0$ in which case the seller still has the bid $\frac{5}{3}$ from bidder $1$, and with probability $\frac{1}{2}$, bidder $2$ bids $2$, in which case the seller makes a strict improvement. \hfill \qed
\end{ex}

Of course, in equilibrium, the deviation in \Cref{ex:fpa} will not increase the auctioneer's profit because bidder $1$ will anticipate the auctioneer using his bid as a reserve for bidder $2$ and so reduce his bid to begin with, eventually leading to a lower profit for the auctioneer due to her inability to commit. Our result differs from that in \citet{akbarpour2020credible} because they assume that the optimal reserve price is commonly known by all bidders, and hence our deviation would be detectable by all bidders. In contrast, in our model, the seller has private information about her cost, so such deviations are not detectable.\footnote{In this example, because types are discrete, a first-price auction with a reserve price of $2$ for both $\theta_0$ types yields the same expected payoff as the auction we described. That first-price auction is credible precisely because the reserve price is cost-independent. When the distribution of types is sufficiently rich, the optimal auctions always require cost-dependent reserve prices and that leads to our main impossibility result for first-price auctions.}

The first-price auction is an example of a static mechanism. However, our main result shows that the non-credibility of first-price auctions is not a coincidence: No static optimal mechanism can be credible (\Cref{thm:static}). The simplistic intuition can be understood from \citet{akbarpour2020credible} who show that optimality plus being static, credible, and winner-paying leaves only one candidate, the first-price auction. Applying their result to each auctioneer type, we would conclude that the auctioneer must use an optimal first-price auction type-by-type, but we have just shown that the first-price auction cannot be credible---hence, an impossibility. While logically tempting, this intuition is flawed because the privately informed auctioneer may design mechanisms that would not fully reveal her types but instead pool her types in arbitrary ways (recall the \textit{inscrutability principle} of \citealt{myerson1983mechanism}). When running the argument of \citet{akbarpour2020credible} type by type, the slice of the mechanism need \textit{not} even be incentive-compatible, let alone optimal.

Instead, our proof proceeds by first characterizing the space of all symmetric credible static mechanisms (\Cref{thm:char}) using a generalization of the dynamic deviations in the first-price auction case. We then show that \textit{(i)} it is without loss to focus on these mechanisms and \textit{(ii)} among these mechanisms, none of them can be optimal. Indeed, we show that restricting to symmetric winner-paying protocols, all credible static mechanisms share the following three properties. First, they are all first-price in the sense of \citet{akbarpour2020credible}: When a bidder submits a message to a mechanism, they must know what price they will pay if they win. Second, the message/bid space can be restricted in a way that is independent of the seller’s cost (this is a generalization of the ex-ante set price floor). Third, they allow for the seller to incorporate into the mechanism the realization of her cost only in a minimal way via a walk-away option: If the best offer is below the seller's cost, the seller would keep the object (for example, via a practice known as \textit{secret reserves}). Among such mechanisms, we show that the seller would always prefer to run an actual auction rather than use a posted price (\Cref{prop:auction}), but she can never achieve the optimal revenue with such a mechanism. Indeed, as we show, the seller's optimal bid space no longer has a standard interval structure (even under a regular distribution) but may involve substantial bid restrictions (\Cref{ex:gap}). The seller's credibility concern leads to a pooling of bidder types that would not appear in the full-commitment solution.

Next, we show that the English (ascending) auction can still be used to credibly implement the optimal mechanism (\Cref{thm:english}). This result does not require the English auction to be run in an open-cry manner but only through bilateral communications. The intuition for this result is as follows. Consider our deviation for the first-price auction.  Suppose that the seller knows that bidder $1$ has stayed in the English auction for a while with the clock rising to $b_1$ (higher than the seller's cost). Now, it might be tempting to conclude that the seller should treat $b_1$ as the new cost to set the reserve price for bidder $2$ (assuming that the seller has not called bidder $2$ yet). But that is not optimal. The optimal thing is to ask bidder $2$ if he is willing to beat $b_1$ and then run an auction between him and bidder $1$. That is, if the seller could not go back to bidder $1$ (like in the first-price auction), then she would inflate the winning bid. But here the seller always wants to go back, and even after such a deviation, she wants to continue as if it were an English auction.

These features of the English auction are not necessarily shared by other dynamic mechanisms, e.g., the Dutch (descending) auction. In \cite{akbarpour2020credible}, both the optimal English and Dutch auctions are credible. However, in our setting, the optimal Dutch auction may not be credible (\Cref{ex:dutch}). This is perhaps surprising because in a Dutch auction, our previous deviation for the first-price auction would not work because when bidder $1$ sends a message ``I will take it at the current price'', it is already too late for the seller to use the bid to update the reserve for bidder $2$. It turns out that there is a second deviation that the seller may use to profitably manipulate the reserves: Upon observing the bad news that the previous bidders declined to bid given the Myersonian reserve price, the seller reduces the reserve price for the last bidder. The key here is that the last bidder does not know he is the last one---otherwise, the Myersonian reserve price would continue to be optimal---and the seller claims to the last bidder that the competition is high but the bidder is lucky to face a low-cost seller. Note that this downward deviation of the reserve becomes profitable precisely when the seller cannot use the upward deviation of the reserve (i.e., no previous bidders bid above the original reserve)---these two dynamic deviations together also imply that the first-price auction has a profitable safe deviation for the seller with probability $1$ (\Cref{thm:fpa}). 

Our model assumes that the auctioneer communicates privately and bilaterally with each of the bidders, following \citet{akbarpour2020credible}.  As we discuss in \Cref{subsec:announcement}, another way to escape our impossibility result is to allow for public announcements, where the auctioneer publicly announces what the reserve price is after privately observing her cost. This public announcement would make the optimal first-price auction credible. It is a common institution in practice, perhaps precisely to avoid the credibility issue that we highlight. However, in many cases, public announcements may not be feasible, practical, or without costs. Moreover, as we discuss in \Cref{subsec:strong}, even with public announcements, the first-price auction may still suffer a credibility concern where the auctioneer deviates jointly with a bidder before the public announcement. We show that the ascending auction continues to be credible with respect to this type of deviation and satisfies what we call ``strong credibility'' (see \Cref{subsec:strong}).

\subsection{Related Literature}

We study credible auction design and uncover a series of dynamic deviations by an informed seller. Our deviations imply an impossibility result for static auctions to achieve both optimality and credibility.\footnote{Our characterization of symmetric credible static auctions (\Cref{thm:char}) also provides a rationale for the existence of secret reserves, which are common in practice (\citealt{bajari2003winner}). Unlike existing explanations (e.g., \citealt{vincent1995bidding}; \citealt{li2017hidden}), our model assumes risk-neutral bidders with independent private values---secret reserves arise due to the credibility concerns of an informed seller.} We show that the dynamic ascending auction can continue to credibly implement the optimal mechanism. Our model of credible auctions follows \citet{akbarpour2020credible}, who develop the notion of credibility using an extensive-form game framework.\footnote{The notion of credibility also appears in \citet{dequiedt2015vertical}, though they impose the restriction to revelation mechanisms (e.g., the auctioneer cannot communicate sequentially with
bidders).} In \citet{akbarpour2020credible}, the first-price auction, English auction, and Dutch auction are all credible. Among these auctions, we show that only the English auction can continue to credibly implement the optimal mechanism when the seller is privately informed. Unlike in \citet{akbarpour2020credible}, the deviations we discover involve the manipulation of the timing and reserve prices during the auction. Recently, \citet*{komo2024shill} build on the framework of \citet{akbarpour2020credible} and \citet{li2017obviously} to study shill-proofness and characterize the Dutch auction as the unique ``strongly shill-proof'' auction among the optimal auctions.

We model the privately informed auctioneer following the literature on informed principals (beginning with \citealt{myerson1983mechanism};  \citealt{maskin1990principal,maskin1992principal}). In our setting, if the auctioneer has full commitment power, then the privacy of her information about the cost is irrelevant (\citealt{Myerson_1985}; \citealt{yilankaya1999note}; \citealt{skreta2011informed}; \citealt{mylovanov2014mechanism}). However, we show that under credibility constraints, the seller's private information is no longer irrelevant because it interacts with the potential deviations during the process of the auction. \citet{giovannoni2022pricing} consider a bilateral trade setting where a seller with limited commitment learns about her cost after proposing the contract, and show that the seller can nevertheless obtain the full-commitment payoff. 

More broadly, this paper also connects to the literature on auctions and mechanism design with limited commitment (\citealt{bester2001contracting}; \citealt{Skreta2006}; \citealt*{Liu2019}; \citealt{banchio2021dynamic}). Unlike in the Coase-conjecture settings, we model the seller's limited commitment as deviations during the auction rather than after the auction. Thus, unlike \citet{doval2022mechanism}, we do not have a revelation principle and study instead various forms of dynamic deviations by an informed seller. 

The remainder of the paper proceeds as follows. \Cref{sec:model} presents our model. \Cref{sec:opt-static} shows the non-credibility of the first-price auction. \Cref{sec:Impossibility} presents our impossibility result for all static mechanisms. \Cref{sec:Dynamic} analyzes dynamic mechanisms. \Cref{sec:discussion} discusses public announcements and strong credibility. \Cref{sec:conclusion} concludes. 

\section{Model}\label{sec:model}
An auctioneer wants to sell one good to a set of bidders $N$ (with $|N| > 1$). Each bidder $i$ has a private value $\theta_i \in \Theta_i$, and the auctioneer has a private cost $\theta_0  \in \Theta_0$. We assume that $\Theta_0 = \Theta_i = [0, 1]$ for all $i \in N$. The distributions of the bidder values are symmetric and have full support, with a continuously differentiable CDF denoted by $F$. We assume that the value distribution is regular, i.e., $\theta - \frac{1 - F(\theta)}{f(\theta)}$ is strictly increasing. We also assume that the distribution of the seller's cost has a continuous, strictly positive density on $[0, 1]$. 

A \textit{\textbf{mechanism}} is a map that specifies the extensive-form game the auctioneer would run depending on her private information $\theta_0$: 
\[M: \Theta_0 \rightarrow \mathcal{G} \]
where $\mathcal{G}$ is the set of all finite-depth extensive-form games for the bidders.

Given a mechanism, let $\mathcal{I}_i(\theta_0)$ be the information sets available to bidder $i \in N$, when the auctioneer is of type $\theta_0$. Let 
\[\mathcal{I}_i := \bigcup_{\theta_0 \in \Theta_0} \mathcal{I}_i(\theta_0)\]
be the union of all information sets of bidder $i$ over the auctioneer's type $\theta_0$. 

An \textbf{\textit{interim strategy}} $\sigma_i: \mathcal{I}_i \rightarrow A$ for bidder $i$ specifies an action $a\in A$ for every bidder $i$'s information set $I_i$, satisfying $\sigma_i(I_i) \in A(I_i)$ where $A(I_i)$ is the set of actions available at $I_i$. An \textit{\textbf{ex ante strategy}} $S_i: \Theta_i \rightarrow \Sigma_i$ for bidder $i$ specifies an interim strategy $\sigma_i \in \Sigma_i$ for every type $\theta_i$. A \textit{\textbf{protocol}} $(M, S_N)$, where $S_N := (S_i(\,\cdot\,))_{i\in N}$, is a pair of a mechanism $M$ and a strategy profile $S_N$. A protocol $(M, S_N)$ is \textit{\textbf{Bayes incentive-compatible (BIC)}} if for all $\theta_i \in \Theta_i$
\[S_i(\theta_i) \in \argmax_{\sigma_i} \E_{\theta_{-i}, \theta_0} \Big[u^{M(\theta_0)}_i(\sigma_i, S_{-i}(\theta_{-i}), \theta_i)\Big]\,,\]
where we define 
\[u^G_i(\sigma_i, \sigma_{-i}, \theta_i) := u_i(x^G(\sigma_i, \sigma_{-i}), \theta_i)\,\]
and $u_i(x, \theta_i)$ denotes the utility of agent $i$ with type $\theta_i$ given outcome $x$, and $x^G(\sigma_i, \sigma_{-i})$ denotes the outcome in game $G$ when agents play according to $(\sigma_i, \sigma_{-i})$. A protocol $(M, S_N)$ satisfies \textit{\textbf{voluntary participation}} if for all $i$, there exists $\sigma'_i$ that ensures bidder $i$ does not receive the good and receives a zero net transfer, regardless of $\sigma'_{-i}$. Throughout the paper, we restrict attention to BIC protocols that satisfy voluntary participation. 

Given a protocol $(M, S_N)$, a \textit{\textbf{messaging game}} is generated as follows: The auctioneer can either choose an outcome or choose to send a message $I_i \in \mathcal{I}_i$ to an agent $i \in N$.  Agent $i$ privately observes message $I_i$ and chooses reply $r \in A(I_i)$. The auctioneer privately observes the reply $r$ and repeats. We say that the auctioneer \textit{\textbf{plays by the book}} if for every $\theta_0$, \textit{(i)} the auctioneer of type $\theta_0$ selects game $M(\theta_0)$, and then \textit{(ii)} the auctioneer contacts players according to the prescription given by $M(\theta_0)$.

Let $S_0: \Theta_0 \rightarrow \Sigma_0$ denote the auctioneer's ex-ante strategy and $\sigma_0 \in \Sigma_0$ denote the auctioneer's interim strategy, when playing the above messaging game. 

Given a promised strategy profile $(S_0, S_N)$, an auctioneer's deviation strategy $\hat{S}_0$ is \textit{\textbf{safe}} if, for all agents $i \in N$, and for all type profiles $(\theta_0, \theta_N)$, there exists a pair $(\hat{\theta}_0, \hat{\theta}_{-i})$ such that 
\[o_i\big(\hat{S}_0, S_N, \theta_0, \theta_N\big) = o_i\big(S_0, S_N, \hat{\theta}_0, (\theta_i, \hat{\theta}_{-i})\big)\,,\]
where $o_i$ denotes the \textit{\textbf{observation}} of bidder $i$ given a strategy profile and a type profile. In particular, we assume that each bidder observes whether he wins the object and his own
payment.

A safe deviation $\hat{S}_0$ is \textit{\textbf{profitable}} for auctioneer type $\theta_0$ if   
\[\E_{\theta_N}\Big[u_0(S_0(\theta_0), S_N, \theta_0, \theta_N)\Big] < \E_{\theta_N}\Big[u_0(\hat{S}_0(\theta_0), S_N, \theta_0, \theta_N)\Big]\,,\]
where $u_0$ denotes the auctioneer's utility function given a strategy profile and a type profile. We assume that the auctioneer maximizes the expected profit net of cost. A protocol $(M, S_N)$ is \textit{\textbf{credible}} if there is no profitable safe deviation: For any safe deviation $\hat{S}_0$, and any auctioneer type $\theta_0$, playing-by-the-book strategy $S_0$ yields a weakly higher payoff:  
\[\E_{\theta_N}\Big[u_0(S_0(\theta_0), S_N, \theta_0, \theta_N)\Big] \geq \E_{\theta_N}\Big[u_0(\hat{S}_0(\theta_0), S_N, \theta_0, \theta_N)\Big]\,.\]

An \textit{\textbf{outcome}} in our setting is a winner (if any) and a profile of payments $(y, t_N) \in (N \cup \{0\}) \times \R^{N}_{\geq 0}$. We assume that the transfers are nonnegative, i.e., the auctioneer does not pay the bidders. We allow for randomized mechanisms, in which case the outcome space is given by $\Delta\big((N \cup \{0\}) \times \R^{N}_{\geq 0}\big)$, and the realization is privately observed by the auctioneer.\footnote{A safe deviation is then defined as one where the auctioneer can provide an innocent explanation by misreporting jointly on $(\theta_0, \theta_{-i}, \varepsilon)$, where $\varepsilon$ is the realization of the randomization, to bidder $i$ that keeps bidder $i$'s observation the same. } Given a protocol $(M, S_N)$, we denote its induced allocation rule and transfer rule as $\big(\tilde{y}^{M, S_N}(\,\cdot\,), \tilde{t}^{M, S_N}(\,\cdot\,) \big): \Theta_0 \times \Theta_N \rightarrow \Delta\big((N \cup \{0\}) \times \R^{N}_{\geq 0}\big)$. We suppress the dependency on $(M, S_N)$ and the randomization whenever it is clear. 

Let 
\[\pi(M, S_N, \theta_0) = \E_{\theta_N}\Bigg[\sum_{i\in N} \tilde{t}_i^{M, S_N}(\theta_0, \theta_N) - \1_{\tilde{y}^{M, S_N}(\theta_0, \theta_N) \neq 0}\cdot \theta_0\Bigg]\]
denote the auctioneer's expected profit of protocol $(M, S_N)$. 

A protocol $(M, S_N)$ is \textit{\textbf{optimal}} if for any BIC protocol $(\widehat{M}, \widehat{S}_N)$ that satisfies voluntary participation, we have 
\[\E_{\theta_0}\Big[\pi(M, S_N, \theta_0)\Big] \geq \E_{\theta_0}\Big[\pi(\widehat{M}, \widehat{S}_N, \theta_0)\Big]\,.\]

A protocol $(M, S_N)$ is \textit{\textbf{static}} if, for each $\theta_0$, each bidder $i$ has exactly one information set and is called exactly once before any terminal history.

A protocol $(M, S_N)$ is a \textit{\textbf{first-price auction}} if $(M, S_N)$ is static, and for each $\theta_0$, each bidder $i$ chooses a bid $b_i$ from a set $B_i(\theta_0) \subset \R_{\geq 0}$ or declines to bid, such that:
\begin{itemize}
    \item[1.] Each bidder $i$ pays $b_i$ if he wins and $0$ if he loses. 
    \item[2.] If any bidder places a bid, then some maximal bidder wins the object. Otherwise, no bidder wins. 
\end{itemize}

\subsection{Model's Features}
  
\paragraph{Informed Principal.}\hspace{-2mm}In our setting, if the auctioneer were not to be constrained by credibility, then the privacy of her cost would be irrelevant (\citealt{Myerson_1985}; \citealt{yilankaya1999note}; \citealt{skreta2011informed}; \citealt{mylovanov2014mechanism}). Indeed, an optimal mechanism in our setting simply maps each cost type $\theta_0$ to a corresponding Myersonian optimal auction (e.g., a first-price auction with the cost-dependent Myersonian reserve). Similarly, if the auctioneer were not to have private information, then our model would collapse to the one studied by \citet{akbarpour2020credible}. In this sense, our model highlights the interaction between credibility and an informed principal. 

\paragraph{Credibility as Equilibrium.}\hspace{-2mm}A credible mechanism can be equivalently defined as a ``promise-keeping'' equilibrium. Fix any protocol $(M, S_N)$. Let $S_0$ be the auctioneer's play-by-the-book strategy given the protocol in the messaging game. Let $\text{Safe}(S_0, S_N)$ be the set of safe deviations in the ex-ante strategy space. Then, just as in \citet{akbarpour2020credible}, our definition of credibility has the property that $(M, S_N)$ is credible if and only if the messaging game has a Bayes-Nash equilibrium $(S_0, S_N)$ when the auctioneer is restricted to play in the set $\text{Safe}(S_0, S_N)$. In this messaging game, there is common knowledge about the menu of the possible extensive-form games 
\[\mathcal{M}:=\big\{M(\theta_0)\big\}_{\theta_0}\]
that the auctioneer promises to implement (depending on her type $\theta_0$). There is also common knowledge about the messaging game that the auctioneer would be playing with the bidders, as well as the common knowledge that the auctioneer can choose any strategy in the set $\text{Safe}(S_0, S_N)$. In particular, along the path of play, the type-$\theta_0$ auctioneer can deviate to another type $\theta'_0$ and then provide innocent explanations about other bidders' types $\theta'_{-i}$ to a given bidder $i$. In general, the safe deviations would involve combinations of such double deviations as the messaging game unfolds. 
In an equilibrium of the messaging game, along any path of play, the bidders not only update their beliefs about other bidders' types but also update their beliefs about the auctioneer's type as well.

\paragraph{Non-Credible Mechanisms.}\hspace{-2mm}The above observation also implies a flavor of ``revelation principle'' in the following sense. Suppose that $M: \Theta_0 \rightarrow \mathcal{G}$ is a non-credible mechanism. Consider the messaging game $M^*$ generated by $M$ as described in the model, and allow the auctioneer to choose any feasible strategy when running the messaging game. Suppose that there exists a pure-strategy Bayes-Nash equilibrium $(\tilde{S}_0, \tilde{S}_N)$ for the auctioneer and the bidders when playing this messaging game. Then, we claim that there exists a credible mechanism $\tilde{M}$ that yields the same equilibrium outcome. Indeed, for each auctioneer type $\theta_0$, the messaging-game strategy $\tilde{S}_0(\theta_0)$ defines an extensive-form game $G^{\tilde{S}_0(\theta_0)}$ played by the bidders. Now, define the mechanism $\tilde{M}$ by $\tilde{M}(\theta_0) = G^{\tilde{S}_0(\theta_0)}$ for all $\theta_0$. Consider the protocol $(\tilde{M}, \tilde{S}_N)$. By construction, $\tilde{S}_0$ is the play-by-the-book strategy for $\tilde{M}$. Moreover, the set of safe deviations $\text{Safe}(\tilde{S}_0, \tilde{S}_N)$ consists of strategies defined in the messaging game $\tilde{M}^*$ generated by $\tilde{M}$, but these strategies are feasible in the original messaging game $M^*$ and hence unprofitable for the auctioneer given that the bidders are playing according to $\tilde{S}_N$.

\subsection{Examples of Credible Mechanisms}

To further illustrate our definitions, we describe some examples of credible mechanisms. A simple credible mechanism is a posted-price mechanism in which the seller walks away without trading if the cost is above the fixed price, and otherwise she allocates the object randomly among the buyers who are willing to pay the fixed price. 

A more sophisticated credible mechanism is to run a ``price waterfall'': The seller sequentially contacts the buyers and offers a price at a time which can depend on the cost and differ across buyers, but each buyer only gets contacted once, and if accepted, must get the object. If the seller uses an optimized price waterfall mechanism, then the mechanism must be credible since each buyer accepts the offer if and only if their value is above the price, and thus any deviation by the seller in the messaging game would simply correspond to a price waterfall with a different sequence of prices.\footnote{The buyers' ordering in this sequence may also change, but they are ex ante symmetric so that would not affect the resulting revenue.} 

Many simple auctions are credible. For example, the seller can use a first-price auction with a public reserve combined with a secret reserve: She allocates the object to the maximal bidder if the bidder bids above both the public reserve and her private cost, and otherwise she keeps the object.\footnote{See \Cref{subsec:static} for a formal definition.} The seller can also use dynamic mechanisms. For example, the seller can use an English (ascending) auction or a Dutch (descending) auction with a fixed reserve price and walk away if the winning bid is below the cost. 

Note that the above examples of auctions are all suboptimal since they significantly restrict the flexibility of the auctioneer to condition on her own private information. As we show, it turns out that there is no static mechanism that achieves credibility and optimality at the same time. However, the English auction can be tailored to have cost-dependent reserves to implement the optimal mechanism in a credible way, while the Dutch auction cannot.

\section{Non-Credibility of First-Price Auctions}\label{sec:opt-static}

Our first result shows that, contrary to the case where the auctioneer's cost is publicly observable, the optimal first-price auction is no longer credible. In fact, we will show a stronger result: The seller has a profitable safe deviation with probability $1$---the seller would always want to cheat regardless of the cost type and the bidders' types. 

\begin{theorem}\label{thm:fpa}
In any optimal, first-price auction, the seller has a profitable safe deviation with probability $1$. 
\end{theorem}
\begin{proof}[Proof of \Cref{thm:fpa}]

Suppose that $(M, S_N)$ is an optimal, first-price auction. By optimality, since the value distributions are symmetric and regular, for almost all $\theta_0$, the induced allocation must be symmetric and deterministic almost everywhere. Then, for any bidder $i$, the interim expected payment is pinned down almost everywhere. Thus, the induced bid by bidder $i$ of type $\theta_i$ when the auctioneer is of type $\theta_0$ 
\[b_i(\theta_i; \theta_0)\]
is pinned down up to a measure-zero set. In particular, $b_i(\theta_i; \theta_0)$ can be assumed to be symmetric across bidders, denoted by $b(\theta_i; \theta_0)$. 

Moreover, because the bidder value distributions are symmetric and regular, there exists a continuous, strictly increasing function $r(\theta_0)$ such that 
$b(\theta_i; \theta_0) = 0$ if $\theta_i < r(\theta_0)$
and $b(\theta_i; \theta_0)$ is strictly increasing in $\theta_i$ if $\theta_i \geq r(\theta_0)$.\footnote{Indeed, $r(\theta_0)$ is pinned down by $r - \frac{1 - F(r)}{f(r)} = \theta_0$, which has a unique solution since $F$ is regular.} Note also that we have 
\[b(\theta_i; \hat{\theta}_0) > b(\theta_i; \theta_0)\]
for all $ \hat{\theta}_0 > \theta_0$ and $\theta_i > r(\hat{\theta}_0)$.\footnote{To see this, note that by the Envelope theorem, \[b(\theta_i; \theta_0) = \theta_i - \frac{1}{Q(\theta_i; \theta_0)}\int^{\theta_{i}}_{r(\theta_0)} Q(s; \theta_0) \d s\]
where $Q(s; \theta_0)$ is the interim allocation probability. Under an optimal mechanism, we have that $Q(s;\theta_0) = F(s)^{|N|-1}$ for all $s \geq r(\theta_0)$ and hence the claim follows by noting that $r(\theta_0)$ is strictly increasing in $\theta_0$.}

Let $b^*$ be the maximal bid given by the first $|N|-1$ bidders. With probability $1$, we have either $b^* > r(\theta_0)$ or $b^* < r(\theta_0)$. 

\textbf{Case (A).} Consider the case of $b^* > r(\theta_0)$. Note also that $b^* < 1 = r(1)$. Now, fix type $\theta_0$ and consider type $\theta_0$'s deviation of pretending to be $\hat{\theta}_0 > \theta_0$ such that 
\[r(\hat{\theta}_0) = b^* \]
and asking the last bidder to bid in the set  $B_{|N|}(\hat{\theta}_0)$. The existence of $\hat{\theta}_0$ is guaranteed by continuity of $r(\,\cdot\,)$. We claim that this must be a strict improvement for the auctioneer. 

There are two cases. First, consider case \textit{(i)} where the last bidder has a value $\theta_{|N|}$ that is weakly below $b^*$. Then, the auctioneer would get revenue $b^*$ under this deviation. But, even if the auctioneer plays by the book, she also gets the same revenue $b^*$, since 
\[b(\theta_{|N|}; \theta_0) \leq \theta_{|N|} \leq b^*\,,\]
which means the maximal bid the auctioneer gets when playing by the book is $b^*$. 

Now, consider case \textit{(ii)} where the last bidder has a value $\theta_{|N|}$ that is strictly above $b^*$. Then, by construction, we have
\[\theta_{|N|} >  r(\hat{\theta}_0)\,.\]
But then the last bidder would bid even more under the auctioneer's deviation since 
\[b(\theta_{|N|}; \hat{\theta}_0) > b(\theta_{|N|}; \theta_0)\,.\]
Since case \textit{(ii)} happens with a positive probability, the safe deviation is a strict improvement for the auctioneer.

\textbf{Case (B).} Consider the case of $b^* < r(\theta_0)$. Fix any type $\theta_0 > 0$. We claim that there exists some $\varepsilon > 0$ small enough that the auctioneer of type $\theta_0$ has a strictly profitable deviation by mimicking $\hat{\theta}_0 = \theta_0 - \varepsilon$ when facing the last bidder. 

Since the reserve price $r(\,\cdot\,)$ is a strictly increasing, continuous function of the cost type, it suffices to show that type $\theta_0$ wants to slightly decrease the reserve price to the last bidder. Let $r^*:= r(\theta_0)$ be the Myersonian reserve for type $\theta_0$. Note that 
\[r^* \in \argmax_r \int^1_0 \1_{\theta_i \geq r} (r - \theta_0) f(\theta_i) \d\theta_i\]
which by regularity of $F$ implies that 
\[-(r^*-\theta_0)f(r^*) + (1 - F(r^*)) = 0\,.\]
Note that, under the deviation, when the auctioneer decides what reserve to charge to the last bidder, the above is not the auctioneer's objective function because the last bidder would assume that there are still $|N|$ bidders in total. Instead, the objective is  
\[W(r) := \int_0^1 \1_{\theta_i \geq r} \big(b^N(\theta_i; r) - \theta_0 \big) f(\theta_i) \d \theta_i\,,\]
where the last bidder bids according to $b^N(\,\cdot\,;r)$ which is the bidding function assuming $|N|$ bidders in total. The bidding function, by the Envelope theorem, is given by
\[b^N(\theta_i; r) = \theta_i - \frac{1}{Q_i(\theta_i; r)}\int^{\theta_{i}}_{r} Q_i(s; r) \d s\,\]
for all $\theta_i \geq r$, where we have $Q_i(\theta_i;r) = F(\theta_i)^{|N|-1}$ given the optimality of $(M, S_N)$. Thus, 
\[W(r) =  \int_0^1 \1_{\theta_i \geq r}\Big(\theta_i - \frac{1}{F(\theta_i)^{|N|-1}} \int_r^{\theta_i} F(s)^{|N|-1} \d s - \theta_0\Big) f(\theta_i) \d \theta_i\,.\]
Therefore, 
\[W'(r) = - (r - \theta_0) f(r) + \int^1_r f(\theta_i)\underbrace{\Big(\frac{F(r)}{F(\theta_i)}\Big)^{|N|-1}}_{< 1} \d\theta_i < -(r-\theta_0) f(r) + (1 - F(r))\,. \]
Thus, 
\[W'(r^*) < -(r^*-\theta_0) f(r^*) + (1 - F(r^*)) = 0\,. \]
Thus, the type $\theta_0$ auctioneer has a profitable deviation by slightly decreasing the reserve to the last bidder (by mimicking a slightly lower $\hat{\theta}_0$).
\end{proof}

\begin{rmk}
In the proof of \Cref{thm:fpa}, the profitable deviation by the seller is dynamic in the sense that it depends on the information revealed during the process of the auction, even though to each bidder, the deviation is \textit{indistinguishable from a static auction}. This can happen precisely because we allow the seller to have private information, which creates uncertainty over the possible reserves the bidders may face even in a static mechanism. 
\end{rmk}

\section{The Impossibility Theorem}\label{sec:Impossibility}

The first-price auctions are a particular class of static mechanisms. However, our second result, building on \Cref{thm:fpa}, shows a completely general impossibility result: There is in fact no static mechanism that is optimal and credible, once the seller has private information. 
\begin{theorem}\label{thm:static}
There exists \emph{no} mechanism that is  
\begin{itemize}
    \item[(i)] credible,
    \item[(ii)] static, and   
    \item[(iii)] optimal\,.
\end{itemize}
\end{theorem}

We defer the proof of \Cref{thm:static} to \Cref{subsec:proof}. The basic intuition behind \Cref{thm:static} can be understood as follows. With private costs, a high-cost seller always has the temptation of walking away from the auction to collect payments from the losing bidders---hence, credible auctions must be winner-paying. But we also know from \citet{akbarpour2020credible} that any optimal, winner-paying, static auctions suffer from the deviations of manipulating the transfers, e.g., misrepresenting the second-highest bid in the second-price auction, unless it is a first-price auction. But we have just shown in \Cref{thm:fpa} that the optimal first-price auctions are not credible due to the dynamic deviations of manipulating the reserves---hence, an impossibility. 

However, this intuition is incomplete because the mechanism $(M(\theta_0), S_N\mid_{\mathcal{I}_N(\theta_0)})$ used by a given type $\theta_0$ of the auctioneer need \textit{not} be optimal or even incentive-compatible if the type $\theta_0$ were to be revealed to the bidders. Indeed, by the \textit{\textbf{inscrutability principle}} of \citet{myerson1983mechanism}, the auctioneer can always use a mechanism that does not reveal her type at all when the bidders take their actions. The deviations identified in \Cref{thm:fpa} would not work directly in such mechanisms because they rely on influencing the last bidder's actions.  \Cref{thm:static} asserts that regardless of how the auctioneer may hide her private information, such a mechanism must either be suboptimal or introduce a profitable safe deviation. The actual proof proceeds in two steps. First, we fully characterize the space of symmetric, winner-paying, credible, and static mechanisms by generalizing our deviations in \Cref{thm:fpa} and the deviations in \citet{akbarpour2020credible}. Specifically, we show that these properties together imply the mechanism must be outcome-equivalent to a first-price auction with a \textit{\textbf{public}} bid space and a \textit{\textbf{walk-away}} option for the auctioneer. Second, we show that any optimal, credible, static mechanism must satisfy winner-paying and symmetry in our setting, and hence must be outcome-equivalent to such a mechanism, but any such mechanism cannot obtain the optimal revenue---hence, an impossibility. 

The rest of this section proceeds as follows. \Cref{subsec:static} characterizes the space of symmetric, winner-paying, credible, and static mechanisms. \Cref{subsec:proof} completes the proof of \Cref{thm:static}. \Cref{subsec:max} discusses how the auctioneer would maximize profits among the symmetric credible static mechanisms.

\subsection{Credible Static Mechanisms}\label{subsec:static}

Since we will allow for suboptimal mechanisms here, we will be more explicit about randomized mechanisms. In particular, let $\varepsilon \in [0, 1]$ be a randomization device and write $\tilde{y}(\theta_N, \theta_0, \varepsilon), \tilde{t}(\theta_N, \theta_0, \varepsilon)$ as the realized allocation and transfer rules. 

A protocol $(M, S_N)$ is a \textit{\textbf{pay-as-bid}} auction if, for each bidder $i$, there exists a function $b_i(\theta_i, \theta_0)$ such that almost everywhere in $\Theta_N \times \Theta_0 \times [0, 1]$, if $\tilde{y}(\theta_N, \theta_0, \varepsilon) = i$, then $\tilde{t}_i(\theta_N, \theta_0, \varepsilon) = b_i(\theta_i, \theta_0)$. 

\begin{lemma}\label{lem:psb}
Every static, credible mechanism must be a pay-as-bid auction. 
\end{lemma}
\begin{proof}[Proof of \Cref{lem:psb}]
Note that for each $\theta_0$, $\Big(M(\theta_0), S_N\mid_{\mathcal{I}_N(\theta_0)}\Big)$ must also be static and credible. Now, conditional on $\theta_0$, we apply Theorem 1 of \citet{akbarpour2020credible} as follows: In particular, for each bidder $i$ and type $\theta_i$, we define $\tilde{\theta}_{-i}:= (\theta_{-i}, \varepsilon)$ as an auxiliary ``opponent type profile.'' Since the randomization $\varepsilon$ is realized after the auctioneer chooses the game, each bidder $i$ observes no additional information about $\varepsilon$ beyond his outcome (whether he wins and his own payment), just like for the opponent type profile $\theta_{-i}$. Therefore, the result follows by applying Theorem 1 of \citet{akbarpour2020credible} to the auxiliary opponent type profile, for each auctioneer type $\theta_0$. 
\end{proof}

Note that a pay-as-bid auction can still have outcomes that are dependent on the auctioneer's types in arbitrary ways. However, we now show that for symmetric and winner-paying mechanisms, the credibility concern due to the seller's own private information will constrain such dependency in a stark way. 

A protocol $(M, S_N)$ is \textit{\textbf{symmetric}} if, for each $\theta_0$, the induced allocation probability
\[q(\theta_N, \theta_0) := \Big(q_i(\theta_N, \theta_0)\Big)_{i \in N} := \Big(\P\big(\tilde{y}(\theta_N, \theta_0, \varepsilon) = i\big) \Big)_{i \in N}\]
is symmetric (i.e., permutation-invariant) across the bidders and the induced payment rule is also symmetric across the bidders.\footnote{The notion of symmetry here follows  \citet{manelli2010bayesian}. The symmetry assumption on the induced payment rule can be relaxed to hold only on the interim payment rule averaging over other bidders' types (but not averaging over the seller's types).} A protocol $(M, S_N)$  is \textit{\textbf{winner-paying}} if almost everywhere in $\Theta_N \times \Theta_0 \times [0, 1]$, if $\tilde{t}_i(\theta_N, \theta_0, \varepsilon) \neq 0$, then $\tilde{y}(\theta_N, \theta_0, \varepsilon) = i$.

A protocol $(M, S_N)$ is a \textit{\textbf{first-price auction with a walk-away option}} if $(M, S_N)$ is static, and for each $\theta_0$, each bidder $i$ chooses a bid $b_i$ from a set $B_i(\theta_0) \subset \R_{\geq 0}$ or declines to bid, such that:
\begin{itemize}
    \item[1.] Each bidder $i$ pays $b_i$ if he wins and $0$ if he loses. 
    \item[2.] If any bidder places a bid, then either some maximal bidder wins the object or the object is kept by the auctioneer. Otherwise, no bidder wins. 
\end{itemize}

A first-price auction with a walk-away option $(M, S_N)$ has a \textit{\textbf{public bid space}} $B \subset \R_{\geq 0}$ if there exists $B \subset \R_{\geq 0}$ such that $B_i(\theta_0) = B$ for all bidders $i$ and all auctioneer types $\theta_0$. 

Two protocols $(M, S_N)$  and $(M', S'_N)$ are \textit{\textbf{outcome-equivalent}} if they induce the same ex-post allocation and transfer rules, almost everywhere in $\Theta_N \times \Theta_0$.

\begin{theorem}\label{thm:char}
Every symmetric, winner-paying, static, and credible $(M, S_N)$ is outcome-equivalent to a first-price auction with a walk-away option and a public bid space, in which the seller walks away if and only if the maximal bid is weakly lower than the cost. 
\end{theorem}

Under symmetry and winner-paying, \Cref{thm:char} shows that, once the seller is constrained by credibility, any static mechanism cannot depend on the seller's cost in any \textit{ex ante} way before communicating with the bidder. Indeed, the bid space is public among all bidders, and the mechanism depends on the cost only via the \textit{ex post} comparison of the seller's cost and the maximal bid. 

The proof of \Cref{thm:char} is in the appendix. To illustrate the intuition, we first sketch the key argument for why the bid spaces $B(\theta_0)$ do not depend on $\theta_0$, assuming that the mechanism $(M , S_N)$ is a first-price auction where the bid spaces $B(\theta_0)$ are discrete, and that the seller never walks away.

To show a public bid space, it suffices to show that, regardless of the auctioneer's type $\theta_0$,  in the bidding game that type-$\theta_0$ auctioneer runs, the symmetric bidding function $b(\theta_i; \theta_0)$ does not depend on $\theta_0$. Because the bidding function is monotone in $\theta_i$, it then suffices to show that the bid distribution $G(\theta_0)$ does not depend on $\theta_0$. 

The key idea is to prove by \textit{induction from the top} (see \Cref{fig:induction}). To illustrate, suppose that we have two auctioneer types $\theta_0$ and $\hat{\theta}_0$. Let $G$ and $\hat{G}$ denote the two bid distributions associated with $\theta_0$ and $\hat{\theta}_0$, respectively. Let $B$ and $\hat{B}$ be the supports of $G$ and $\hat{G}$, respectively. We assume that $B$ and $\hat{B}$ are finite, and write $B = \{b^1, \dots, b^n\}$, where $b^1 > \cdots > b^n$, and $\hat{B} = \{\hat{b}^1, \dots, \hat{b}^m\}$, where $\hat{b}^1 > \cdots > \hat{b}^m$. We further denote the probability mass of $G$ on $b^k$ by $\mu^k$, and similarly the probability mass of $\hat{G}$ on $\hat{b}^k$ by $\hat{\mu}^k$.

\begin{figure}[t]
    \centering
    \begin{subfigure}[b]{0.45\textwidth}
        \centering
        \includegraphics[width=\textwidth]{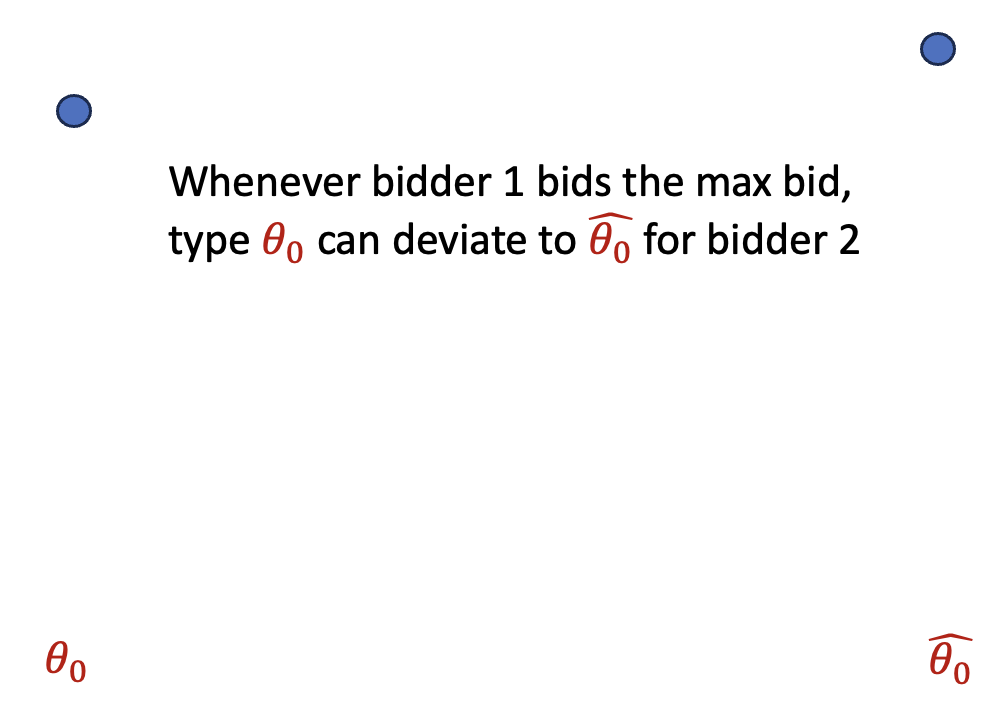}
        \caption{\text{ }}
        \label{fig:panelA}
    \end{subfigure}
    \hfill
    \begin{subfigure}[b]{0.46\textwidth}
        \centering
        \includegraphics[width=\textwidth]{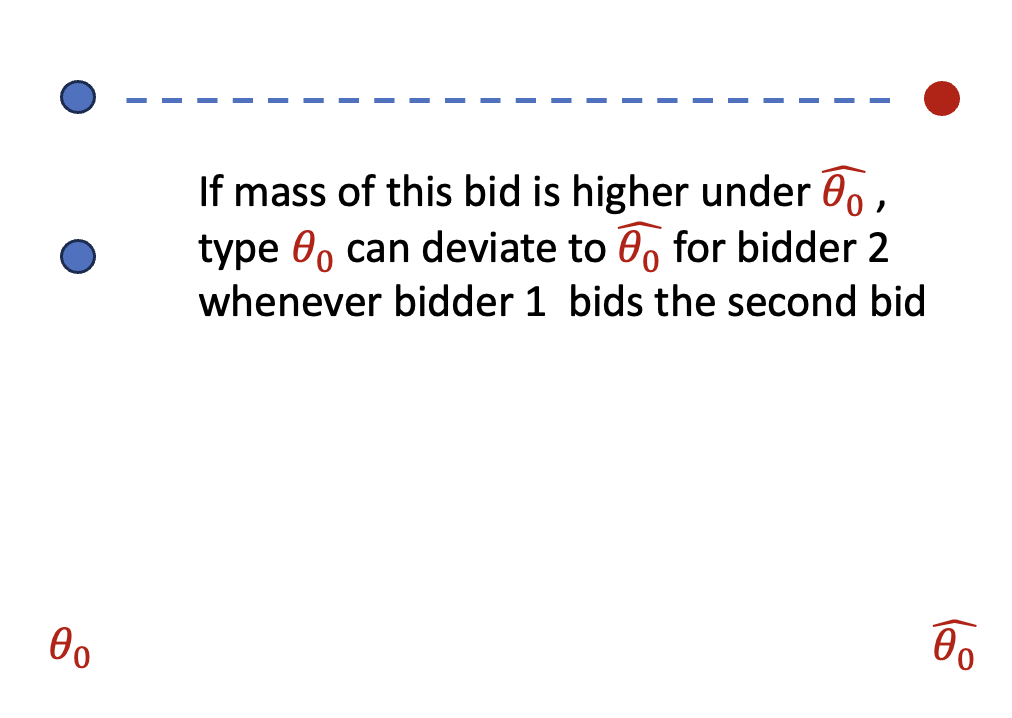}
        \caption{\text{ }}
        \label{fig:panelB}
    \end{subfigure}
    
    \vskip\baselineskip
    \begin{subfigure}[b]{0.45\textwidth}
        \centering
        \includegraphics[width=\textwidth]{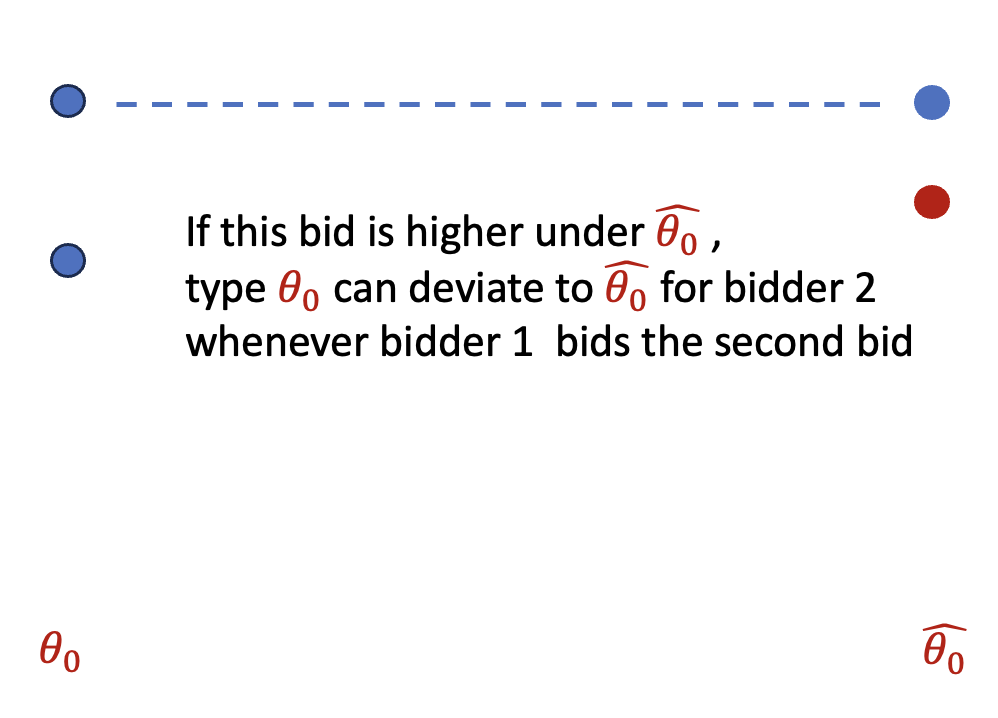}
        \caption{\text{ }}
        \label{fig:panelC}
    \end{subfigure}
    \hfill
    \begin{subfigure}[b]{0.46\textwidth}
        \centering
        \includegraphics[width=\textwidth]{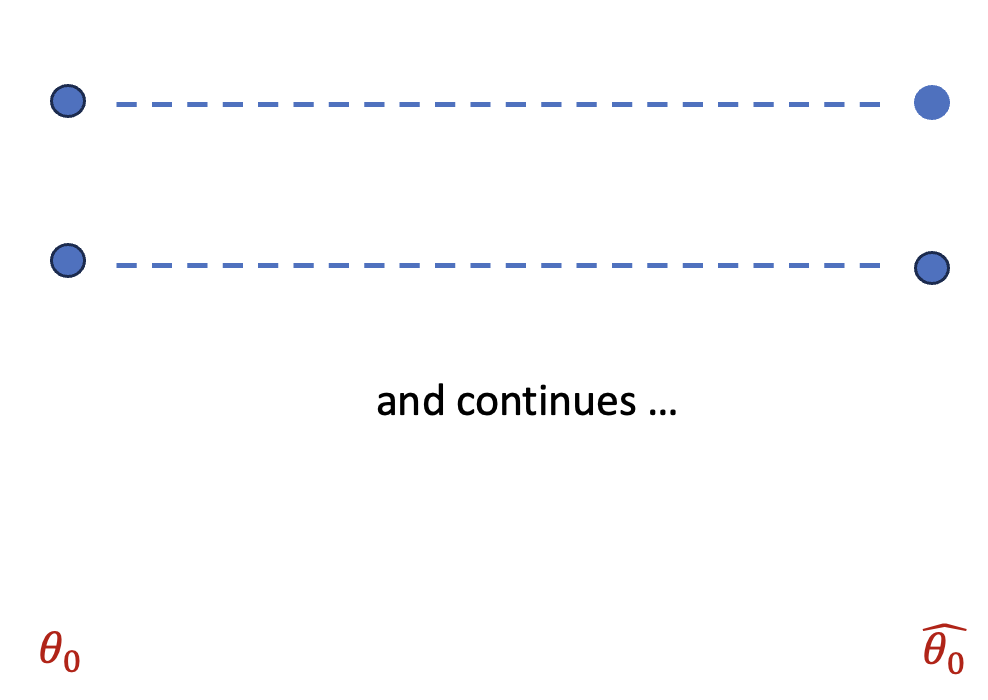}
        \caption{\text{ }}
        \label{fig:panelD}
    \end{subfigure}
    
    \caption{Illustration of the induction argument for the proof of \Cref{thm:char}}
    \label{fig:induction}
\end{figure}

First, we claim that the top bid must be the same, i.e., $b^1 = \hat{b}^1$. Suppose for contradiction that this is not the case, e.g., $\hat{b}^1 > b^1$. Then, note that type-$\theta_0$ auctioneer has a profitable safe deviation: First, elicit the bids from $|N|-1$ bidders, and then mimic type $\hat{\theta}_0$ when facing the last bidder if the highest bid the auctioneer receives from the $|N|-1$ bidders is already $b^1$. Since this contradicts credibility, we have $b^1 = \hat{b}^1$. \Cref{fig:panelA} illustrates.

Second, we claim that the probability mass on the top bid must also be the same, i.e., $\mu^1 = \hat{\mu}^1$. Suppose for contradiction that this is not the case, e.g., $\hat{\mu}^1 > \mu^1$. Then, by the previous step, type-$\theta_0$ auctioneer has the following profitable safe deviation: First, elicit the bids from $|N|-1$ bidders, and then mimic type $\hat{\theta}_0$ when facing the last bidder if the highest bid the auctioneer receives from the $|N|-1$ bidders is $b^2$. Indeed, if playing by the book, type-$\theta_0$ auctioneer has only a probability of $\mu^1$ to increase the highest bid from $b^2$ to $b^1$, but following this deviation, the auctioneer is strictly better off in expectation.  \Cref{fig:panelB} illustrates.

Now, we claim that $b^2 = \hat{b}^2$. Indeed, if this is not the case, e.g., if $\hat{b}^2 > b^2$, then by the previous step, type-$\theta_0$ auctioneer has again a profitable safe deviation as before: First, elicit the bids from $|N|-1$ bidders, and then mimic type $\hat{\theta}_0$ when facing the last bidder if the highest bid the auctioneer receives from the $|N|-1$ bidders is $b^2$. This then implies that we must also have $\mu^2 = \hat{\mu}^2$. Indeed, if this is not the case, e.g., if $\hat{\mu}^2 > \mu^2$, then type-$\theta_0$ auctioneer also has a profitable safe deviation: First, elicit the bids from $|N|-1$ bidders, and then mimic type $\hat{\theta}_0$ when facing the last bidder if the highest bid the auctioneer receives from the $|N|-1$ bidders is $b^3$. \Cref{fig:panelC} illustrates.

Then, the induction argument continues and we must have $b^k = \hat{b}^k$ for all the bids in the support, and $\mu^k = \hat{\mu}^k$ for all the probability mass associated with the bids. Thus, the two bid distributions $G$ and $\hat{G}$ must be identical.  \Cref{fig:panelD} illustrates.

\begin{rmk}
The actual proof follows this argument closely, after showing that the mechanism $(M, S_N)$ must be outcome-equivalent to a first-price auction where the seller can walk away. In addition, the actual proof generalizes the above induction argument to \textit{(i)} allow arbitrary bid distributions, which can also be continuous, and to \textit{(ii)} incorporate that the auctioneer will walk away when the maximal bid is below her cost. 
\end{rmk}

\subsection{Proof of the Impossibility Theorem}\label{subsec:proof}

Now we prove \Cref{thm:static}. Suppose for contradiction that $(M, S_N)$ is an optimal, credible, and static mechanism. As in the proof of \Cref{thm:fpa}, by optimality, since the value distributions are symmetric and regular, the induced allocation rule must be deterministic almost everywhere. Moreover, by credibility, the induced transfer rule $\tilde{t}(\theta_N, \theta_0)$ must also be deterministic almost everywhere, and there exists some $b_i(\theta_i, \theta_0)$ such that  $\tilde{t}_i(\theta_N, \theta_0) = b_i(\theta_i, \theta_0)$ for almost all $\theta_N$ and $\theta_0$ such that $\tilde{y}(\theta_N, \theta_0) = i$ by \Cref{lem:psb}. 

\textbf{Step 1.} We first show that any credible static optimal mechanism must be \textit{\textbf{winner-paying}} in our setting, i.e., for almost all $\theta_N$ and $\theta_0$, if $\tilde{t}_i(\theta_N, \theta_0) \neq 0$, then $\tilde{y}(\theta_N, \theta_0) = i$. 

\begin{lemma}\label{lem:winner-paying}
Every credible static optimal mechanism must be winner-paying. 
\end{lemma}

The proof of \Cref{lem:winner-paying} is in the appendix. We sketch the proof here. First, we show that any type $\theta_0$ of the seller can obtain at most $1 - \theta_0$ in any such mechanism. Indeed, by credibility, the mechanism must be incentive-compatible for the type-$\theta_0$ seller to follow the prescribed strategy instead of the strategy of type $\hat{\theta}_0$. Therefore, by the Envelope theorem, given that the allocation rule must be the Myersonian allocation rule, the expected payoff of each type-$\theta_0$ seller is pinned down up to the payoff of the highest cost type. However, under any optimal, first-price auction, we have that \textit{(i)} the payoff of the highest cost type is $0$, and \textit{(ii)} the seller of type $\theta_0$ also has no incentive to mimic another type $\hat{\theta}_0$ (at the beginning) given that the reserves are already set optimally.\footnote{Indeed, the profitable deviation in \Cref{thm:fpa} only happens during the process of the auction as the seller learns, not at the beginning.} Therefore, any optimal, credible mechanism must yield the seller of the highest cost type $0$ payoff, which then implies that any cost type $\theta_0$ can obtain a payoff of at most $1 - \theta_0$. 

Now suppose for contradiction that there exists an optimal, credible, and static mechanism that is not winner-paying. In particular, at some cost type $\hat{\theta}_0$, the expected payment from the losing bidders must be strictly positive. Then, the auctioneer of sufficiently high cost type $\theta_0$, which by the above observation implies a vanishing payoff, must have an incentive to deviate to the cost type $\hat{\theta}_0$ to collect the payments from the losing bidders, and then claim to every bidder that they are outbid by someone else and keep the object herself. A contradiction. 

\textbf{Step 2.} Now, we know that $(M, S_N)$ must be credible, static, optimal, winner-paying, and pay-as-bid. We next show that $(M, S_N)$ must be a \textit{\textbf{symmetric}} protocol as defined in \Cref{subsec:static}. In particular, we show that for any cost type $\theta_0$, the ex-post payment rule of $(M, S_N)$ must be symmetric across the bidders. 

\begin{lemma}\label{lem:ex-post}
Every credible static optimal mechanism must be symmetric. 
\end{lemma}

The proof of \Cref{lem:ex-post} is in the appendix. We sketch the proof here. As argued before, optimality already implies a symmetric ex-post allocation rule for the bidders. Thus, it suffices to show that for each $\theta_0$, the ex-post payment rule is also symmetric across bidders.\footnote{Note that the Envelope theorem would imply the interim payment rule averaging over both other bidders' types \textit{and} the seller's types is symmetric across bidders. However, this would not be sufficient for our purposes since it averages over the seller's types.} The proof leverages credibility in multiple ways. First, because of credibility, as in the proof of \Cref{lem:psb}, the payment $b_i$ of any bidder $i$ upon winning can only depend on their type $\theta_i$ and the information set $I_i$ they receive (since the mechanism is static, recall that each bidder only receives one information set but that information set can depend on the seller's type $\theta_0$). Now using the BIC constraints for the bidder after receiving information set $I_i$, this implies that the bidder's ``bid'' $b_i(\theta_i, I_i)$ must be non-decreasing and continuous in $\theta_i$ by the Envelope theorem (given the Myersonian allocation rule by optimality). As a notation, write $b_i(\theta_i, \theta_0) = b_i(\theta_i, I_i(\theta_0))$. 

Credibility also implies that at every event that bidder $i$ gets the object, we must have 
\begin{equation}
    b_i(\theta_i, \theta_0) \geq \max\big\{b_j(\theta_j, \theta_0), \theta_0\big\}\,,\label{eq:max}
\end{equation}
for all other bidders $j \neq i$.  Indeed, suppose not. Then, either $\theta_0 > b_i(\theta_i, \theta_0)$ or $b_j(\theta_j, \theta_0) > b_i(\theta_i, \theta_0)$. In the first case, the seller can keep the object and claim to every bidder that there exists another bidder with a strictly higher type (recall that we have the Myersonian allocation by optimality). In the second case, the seller can allocate the object to bidder $j$, and claim to every bidder that bidder $j$ has the highest type. By credibility, neither deviation can be profitable and hence we have \eqref{eq:max}. 

For any $s \in [0, 1]$, if $s < r(\theta_0)$, then the Myersonian allocation implies that type-$s$ of bidder $i$ has zero probability of getting the object, and hence without loss $b_i(s, \theta_0)$ can be set to $0$. Now, we claim that for any two bidders $i, j$, we have 
\begin{equation}
    \max\big\{b_i(s, \theta_0), \theta_0\big\}  = \max\big\{b_j(s, \theta_0), \theta_0\big\} \,. \label{eq:sym}
\end{equation}
Indeed, suppose not. Then, without loss, we have $b_i(s, \theta_0) >  \max\{b_j(s, \theta_0), \theta_0\}$ at some $(s, \theta_0)$. Since $b_i(s, \theta_0) > \theta_0$, by construction, we have that $s \geq r(\theta_0)$. By continuity of $b_j(\,\cdot\,, \theta_0)$, there exists some sufficiently small $\varepsilon > 0$ such that 
\[b_i(s, \theta_0) >  \max\Big\{b_j(s + \varepsilon, \theta_0), \theta_0\Big\}\,.\]
Consider the type profile where bidder $i$ has type $s$, bidder $j$ has type $s + \varepsilon$, and all other bidders have type $0$. Then, given $r(\theta_0) \leq s < s+ \varepsilon$, the Myersonian allocation rule must allocate the object to bidder $j$, but that contradicts credibility by \eqref{eq:max}. 

Finally, for any $s \in [0, 1]$, if $s \geq r(\theta_0)$, then we have $b_i(s, \theta_0) \geq \theta_0$, because otherwise at the type profile $(s, 0, \dots, 0)$ where bidder $i$ has type $s$, the seller would not sell the object by credibility, contradicting the Myersonian allocation. Then, it follows immediately from \eqref{eq:sym} that for any bidders $i$ and $j$, and for any $s \in [0, 1]$, we have $b_i(s, \theta_0) = b_j(s, \theta_0)$, concluding the proof of \Cref{lem:ex-post}.  

\textbf{Step 3.} Given \Cref{lem:winner-paying} and \Cref{lem:ex-post}, we can now apply \Cref{thm:char} to finish the proof of \Cref{thm:static}. Indeed, applying \Cref{thm:char} shows that there exists a mechanism $(M', S'_N)$ that is a first-price auction with a walk-away option and a public bid space such that $(M', S'_N)$ is outcome-equivalent to $(M, S_N)$. In particular, $(M', S'_N)$ must be optimal. We argue for a contradiction.

Let $Q_i(\theta_i)$ be the symmetric interim allocation rule (averaging over the type profiles of bidders $j \neq i$ and the seller's cost types). Since the ex post allocation rule coincides with the Myersonian allocation rule, we must have $Q_i(\theta_i) = 0$ for all $\theta_i < r(0)$ and $Q_i(\theta_i) > 0$ for all $\theta_i > r(0)$. Then, given that $(M', S'_N)$ is optimal and has a public bid space, the Envelope theorem implies that there exists a symmetric bidding function $b(\theta_i)$ (independent of the cost type $\theta_0$) such that $b(\theta_i) > 0$ for all $\theta_i > r(0)$, and $b(\theta_i) = 0$ for all $\theta_i < r(0)$ (i.e., these types decline to bid). Indeed, the Envelope theorem implies that for any $\theta_i > r(0)$, we have
\[\theta_i - b(\theta_i) = \frac{\int^{\theta_i}_{r(0)} Q_i(s) \d s}{Q_i(\theta_i)}\leq \frac{(\theta_i - r(0)) Q_i(\theta_i) }{Q_i(\theta_i)} = \theta_i - r(0)\,, \]
where the inequality uses that $Q_i(\,\cdot\,)$ is non-decreasing. Therefore, for any $\theta_i > r(0)$, 
\[b(\theta_i) \geq r(0) > 0\,.\]
Then, for any cost type $\theta_0$ such that 
\[0 < \theta_0 < r(0) \,,\]
whenever there is any bid submitted which must be weakly above $r(0)$, the seller of type $\theta_0$ cannot walk away by credibility. Therefore, for any cost types $\theta_0, \theta'_0 \in [0, r(0))$, we have that the expected probability of trade (averaging over all bidder types) must coincide: 
\[\P_{\theta_N}\big(\tilde{y}(\theta_N, \theta_0) \neq 0\big) = \P_{\theta_N}\big(\tilde{y}(\theta_N, \theta'_0) \neq 0\big)\,.\]
But then $(M', S'_N)$ cannot be optimal since this cannot hold under the Myersonian allocation rule---a contradiction. This completes the proof of \Cref{thm:static}.

\subsection{Maximizing Profits among Credible Static Mechanisms}\label{subsec:max}
Given the impossibility theorem (\Cref{thm:static}), we cannot achieve the Myersonian revenue with credible static mechanisms. Now, we explore how the profit-maximizing mechanism would look like when the seller is restricted to credible static mechanisms. Under symmetry and winner-paying, \Cref{thm:char} asserts that, if the auctioneer uses credible static auctions, then the only degree of freedom that the auctioneer has is the public bid space $B$. The walk-away option by the seller can also be thought of as having a secret reserve price. In particular, \Cref{thm:char} reduces the problem of maximizing profits in this class of mechanisms to simply designing the public bid space $B$. 

The next result shows that using an actual auction with a bid space $|B| > 1$ can always generate more profit than using a posted price ($|B| = 1$): 

\begin{prop}\label{prop:auction}
For any $(M, S_N)$ that maximizes the expected profit among symmetric, winner-paying, credible, static mechanisms, $(M, S_N)$ is outcome-equivalent to a first-price auction with a walk-away option and a public bid space $B$ where $B \supseteq \{b_0, b_1\}$ for some $b_0 \neq b_1 \in (0, 1)$. 
\end{prop}

\begin{proof}[Proof of \Cref{prop:auction}]
By \Cref{thm:char}, any such protocol $(M, S_N)$ is outcome-equivalent to a first-price auction with a walk-away option and a public bid space $B \subseteq [0, 1]$ in which the seller walks away if and only if the maximal bid is below the cost.

Suppose that we use a posted price mechanism, i.e., $|B| = 1$, in particular, $B = \{b_0\}$ for some $b_0 \in (0, 1)$. We construct an improvement. Let $\tilde{B}=\{b_0, b_1\}$ for some $b_1 \in (b_0, 1)$. 
For any such $\tilde{B}$, by \citet{reny2011existence}, there exists a symmetric pure-strategy bidding equilibrium $b(\,\cdot\,)$ in the first-price auction with a walk-away option. In particular, for any bidder $i$ and type $\theta_i > b_0$, that type will participate and submit some bid $b(\theta_i) \geq b_0$. Moreover, note that for $b_1$ sufficiently close to $b_0$, it must be that for a positive measure of types $\theta_i$, we have $b(\theta_i) > b_0$ in any symmetric equilibrium.\footnote{Indeed, it suffices to take any $b_1$ such that $b_0 < b_1 < 1 - (1 - b_0) \alpha$, where $\alpha \in (0, 1)$ is the interim probability of winning for a given bidder when bidding at $b_0$ conditional on the seller's cost type $\theta_0 \leq b_0$ (assuming other bidders following the strategy of biddng at $b_0$ if and only if their types $\theta_i \geq b_0$).} Fix any such $\tilde{B}$. Since the seller has the option to walk away, the expected profit must be increased because realization by realization, the profit is increased: 
\[\1_{\max_i \theta_i \geq b_0} \max\Big\{b_0 - \theta_0, 0\Big\} \leq \1_{\max_i \theta_i \geq b_0} \max\Big\{b\big(\max_i \theta_i\big) - \theta_0, 0\Big\}\,,\]
where the inequality is strict for a positive-measure set of type profiles.  Hence, the original mechanism cannot attain the maximal expected profit. Since this holds for any posted price mechanism, the result follows. 
\end{proof}

However, which public bid space $B$ is optimal for the auctioneer generally depends on the details of the environment. If the auctioneer never walks away upon observing the bids, then it is easy to see that the optimal bid space $B$ is given by an interval $[R^*, \infty)$, where $R^*$ is the Myersonian reserve with respect to the average cost $\E[\theta_0]$.\footnote{To see this, note that for any interim allocation rule where $Q(\theta_i;\theta_0) = Q(\theta_i)$, we can write the expected profit as $|N| \cdot \int_0^1 Q(\theta_i)\big(\text{MR}(\theta_i) - \theta_0\big) \d F(\theta_i) \d G(\theta_0)  = |N| \cdot \int_0^1 Q(\theta_i)\big(\text{MR}(\theta_i) - \E[\theta_0]\big) \d F(\theta_i)$, where $F$ is the CDF for $\theta_i$, $G$ is the CDF for $\theta_0$, and $\text{MR}(\theta_i)$ is the Myersonian virtual value. The regularity of the value distribution then implies that the optimal interim allocation rule is given by the ``assortative matching'' rule up to a threshold type defined by $\text{MR}(\theta_i) = \E[\theta_0]$. } However, perhaps surprisingly, the next example shows that the seller can benefit from restricting bids in a way to induce some pooling over bidder types and then walking away sometimes: 

\begin{ex}\label{ex:gap}
Suppose that we have two bidders and the values $\theta_i \sim U[0, 1]$. Suppose that the seller's cost $\theta_0$ is either $0$ or $0.5$ with equal probabilities. The Myersonian reserve for the average cost type $\E[\theta_0] = 0.25$ is given by $R^* = (1 + 0.25) / 2 = 0.625$. The expected profit for the seller using the first-price auction with the public bid space $[0.625, \infty)$ is 
\[2 \cdot \int_{0.625}^1 \big(\text{MR}(\theta_i) -0.25\big) \cdot Q(\theta_i) \d \theta_i  = 2 \cdot \int_{0.625}^1 \big(2\theta_i - 1.25\big) \cdot \theta_i \d \theta_i \approx 0.246\,.\]
However, consider a first-price auction with a public bid space 
\[B = \{0.5\} \cup [0.625, \infty)\]
where the seller of type $\theta_0 = 0$ never walks away and the seller of type $\theta_0 = 0.5$ walks away for the bid $b = 0.5$. The equilibrium of this auction is characterized by a threshold type $\theta^*$ who is indifferent between bidding $0.5$ and bidding $0.625$, i.e., 
\[\big(\theta^* - 0.5\big)  \cdot \frac{1}{2} \cdot \big(\frac{1}{2} \cdot (\theta^* - 0.5) + 1 \cdot 0.5 \big)  =  \big(\theta^*  - 0.625\big) \cdot 1 \cdot \theta^*\,.\]
Then, we have $\theta^* \approx 0.717$. The expected profit for the seller is then given by 
\[2 \cdot \int_{0.717}^1 \big(\text{MR}(\theta_i) -0.25\big) \cdot \theta_i \d \theta_i + 2 \cdot \frac{1}{2} \cdot \int_{0.5}^{0.717} \big(\text{MR}(\theta_i) - 0\big) \cdot \big( \frac{1}{2} \cdot (0.717 - 0.5) + 1 \cdot 0.5\big) \d \theta_i \approx 0.263\,,\]
which is strictly higher than the seller's profit under the bid space $[0.625, \infty)$. \hfill \qed 
\end{ex}

\begin{rmk}
In \Cref{ex:gap}, the optimal bid space must be restricted to induce some pooling of bidder types: If the bid space is of the form $[R, \infty)$, where $R \leq 0.5$, then the seller can improve it by using $[0.5, \infty)$ (since $\text{MR}(\theta_i) \leq 0$ for all $\theta_i \leq 0.5$), but then it can be further improved to be $[R^*, \infty)$ where $R^* = 0.625$ as argued above. But that is dominated by using the bid space $\{0.5\} \cup [0.625, \infty)$ which induces a bunching of bidders at the bid $0.5$. Intuitively, the seller introduces pooling at a lower bid to allow herself to walk away credibly when her cost is realized to be higher than the maximal virtual value. With full commitment, the seller can always do that using cost-dependent ex ante reserve prices, but as we show, credibility eliminates any ex ante dependency.     
\end{rmk}

\section{Credibility of Dynamic Mechanisms}\label{sec:Dynamic}

Given the impossibility result (\Cref{thm:static}), to maintain credibility, an informed seller must give up the possibility of either having a static auction or an optimal auction. 

Our third result shows that if the auctioneer can use dynamic mechanisms, then she can still obtain the optimal profit using a credible mechanism. In particular, she can do so by running an ascending (English) auction. 

We say that a protocol $(M, S_N)$ is an \textit{\textbf{ascending auction}} if $\Big(M(\theta_0), S_N\mid_{\mathcal{I}_N(\theta_0)}\Big)$ is an ascending auction as defined in \citet{akbarpour2020credible},\footnote{See Definition 14 in \citet{akbarpour2020credible}. For details, see \Cref{app:ascending}. } for each auctioneer type  $\theta_0$. In particular, for each $\theta_0$, the bid space of agent $i$ is just $\Theta_i$, which is assumed to be discrete. We assume that the auctioneer's type space is also discrete. The reserve prices are captured by the initial bid for each agent $i$. 

\begin{theorem}\label{thm:english}
There exists an optimal, credible mechanism. In particular, an optimal, ascending auction is credible. 
\end{theorem}

\begin{proof}[Proof of \Cref{thm:english}]
We use a similar argument as in \citet{akbarpour2020credible}. Note that, in an ascending auction, the misrepresentation of the auctioneer's type and any opponent's type would not change the optimal strategy used by bidder $i$. Specifically, suppose that $(M, S_N)$ is an optimal, ascending auction. Then for each $i$, $S_i$ specifies the strategy of quitting the auction if and only if the running bid is above agent $i$'s type $\theta_i$. Let $S'_0$ be any safe deviation. Note that the bidder's strategy $S_i$ is also a best response to $(S'_0, S_{-i})$.\footnote{See Lemma 3 of \citet{akbarpour2020credible}. In \Cref{app:ascending}, we also formally prove a generalization of their Lemma 3 and use it to show ``strong credibility'' of the optimal ascending auction (see \Cref{subsec:strong}).} 

Now suppose for contradiction that there exists a safe deviation $S'_0$ that is profitable for some auctioneer type $\theta^*_0$. Then consider the auctioneer's strategy $\hat{S}_0$ defined by playing according to $S'_0(\theta^*_0)$ if the auctioneer's type is $\theta^*_0$, and playing according to $S_0(\theta_0)$ otherwise. Clearly, $\hat{S}_0$ is also a safe deviation. Then the strategy profile $(\hat{S}_0, S_N)$ would induce a BIC mechanism as argued above. But this implies that we have found a BIC mechanism that yields a strictly higher expected profit than $(M, S_N)$ when averaging over the seller's types $\theta_0$ (recall that we have finite types), contradicting that $(M, S_N)$ is optimal. 
\end{proof}

\begin{rmk}
While the proof of \Cref{thm:english} is concise, it is perhaps surprising that our deviation for the first-price auction that involves the manipulation of the reserve prices no longer works for the ascending auctions. To understand the intuition, suppose that we have two bidders. Suppose that the seller knows that bidder $1$ has stayed in the ascending auction for a while with the clock rising to some $b_1$ higher than the seller's cost. Even though it might be tempting to conclude that the seller should treat $b_1$ as the new cost to set the reserve price for bidder $2$ (assuming that the seller has not called bidder $2$ yet), it is \textit{not} optimal to do so. This is because the seller can go back to bidder $1$ in the ascending auction. The optimal thing is to ask bidder $2$ if he is willing to beat $b_1$ and then run an auction between him and bidder $1$. If the seller could not go back to bidder $1$, as in the first-price auction, then she would have reason to inflate the winning bid. But here the seller always has an incentive to go back, and even after such a deviation, she intends to continue \textit{as if it were an ascending auction}. 
\end{rmk}

As the above discussion suggests, the credibility of the English auction in our setting relies on a few key features of the English auction, which may not be shared by other dynamic mechanisms, e.g., the Dutch (descending) auction. The next example shows that, contrary to the case where the auctioneer's cost is publicly observable, the optimal Dutch auction may not be credible: 
\begin{ex}[Optimal Dutch auction is not credible]\label{ex:dutch}There are two cost types $c_{1}=0,c_{2}=0.7$. There are two value types $v_{1}=1,v_{2}=2$ with equal probabilities. The Myersonian reserves are $r(c_{1})=1,r(c_{2})=2$. For a fixed set of bidders $N$, the bidding function is 
\[
\quad b(v_{1};r=1)=1,\qquad b(v_{2};r=1)=k_{|N|}
\]
\[
b(v_{1};r=2)=0,\qquad b(v_{2};r=2)=2
\]
where $k_{|N|}\rightarrow2$ as $|N|\rightarrow\infty$. Fix any $n$ such that $k_{n}\geq 1.8$.  Suppose that we have $n$ bidders. Consider the cost type $c_{2}$. Consider the event that the first $n-1$ bidders stay out of the
auction given the reserve $2$. Now, consider the safe deviation that,
when facing the last bidder, the auctioneer pretends to be of the cost type
$c_{1}$ and sets the reserve price $r=1$. The last bidder will bid $1$
if of low type and $k_{n}$ if of high type, which means that the auctioneer
gets a profit 
\[
1\times\frac{1}{2}+k_{n}\times\frac{1}{2}-0.7\geq1.4-0.7=0.7\,.
\]
If playing by the book and setting $r=2$ for the last bidder, the
auctioneer only gets $\frac{1}{2}\times(2-0.7)=0.65<0.7$. Hence, this
deviation is profitable, and the Dutch auction is not credible. \hfill \qed
\end{ex}

\begin{rmk}
The deviation in \Cref{ex:dutch} is exactly the downward deviation used in the proof of \Cref{thm:fpa}. Indeed, our upward deviation for the first-price auction cannot work for the Dutch auction because the clock is running down---when bidder $1$ sends a message ``I will take it at the current price'', it is already too late for the seller to use the bid to update the reserve for bidder $2$. Instead, the deviation in \Cref{ex:dutch} considers the event where the previous bidders have values lower than the Myersonian reserve, and then reduces the reserve faced by the last bidder. Importantly, as in the proof of \Cref{thm:fpa}, the last bidder does not know that he is the only one left---otherwise, the seller would optimally set the same reserve price regardless of the number of bidders (\citealt{Myerson1981}). In equilibrium, once the bidders anticipate this deviation, they will shade their bids even more, eventually leading to a profit loss for the seller due to her inability to commit. 
\end{rmk}

\section{Discussion}\label{sec:discussion}

\subsection{Public Announcements}\label{subsec:announcement}

So far we have assumed that the seller communicates with each of the bidders privately. One of the institutions used by auctioneers in some markets is public communication and such institutions could help with the credibility of auctions as we now discuss.

First, an obvious solution to the problem of credibility of the first-price auction with an informed seller is for the seller to commit to announcing publicly the reserve price, before communicating with any of the bidders. This kind of public auction would make first-price auctions with optimal cost-dependent reserves credible. 

Second, even if the seller could not announce publicly the reserve, even just announcing publicly the moment all the bids have to be submitted could make the first-price auction credible. For example, the seller would first publicly announce that they will only accept bids on Sunday at noon (as is the case in some real estate sales). Then, the seller could communicate privately with every bidder the reserve price, without collecting any information about bidder valuations or their willingness to bid above the reserve. Since the seller would not collect responses from bidders before Sunday’s noon auction, there would be no safe profitable deviation for the seller from the optimal first-price auction. 

Third, public announcements could make many other auction formats credible as well. For example, in the case of the sealed-bid second-price auction, public communication in the form of revealing after the auction all bids and the identities of the bidders could make the second-price auction credible. (This is a common practice in the spectrum auctions run by the FCC in the US and by Industry Canada in Canada---there, while identities of bidders are private during the auction, all bids and bidders are revealed publicly after the auction, allowing bidders to verify that the complex auction rules, such as computing VCG payments and ``core adjustment'' have been followed.) With such public announcements, unless the auctioneer can “invent” bidders, second-price sealed-bid auctions would be credible. 

In summary, the institution of public announcements that is often used in practice can go a long way toward helping sellers design credible and revenue-maximizing auctions. That said, public announcements of the kind described above are not always feasible, practical, or without cost. First, the seller may prefer to keep the very existence of an auction secret, and contact only a small number of traders. Second, bidders, especially losing bidders, may prefer to have their valuations or even participation in the auction secret. Third, the time when the good becomes available for an auction may be random and privately observed by the seller. Especially when combined with the desire to keep the auction private, this could undermine strategies like “all bids are only accepted at noon.” Finally, public communication can be costly.

\subsection{Strong Credibility of Ascending Auctions}\label{subsec:strong}

So far we have assumed that the only credibility constraint is that the seller does not have any safe deviation that is profitable. We did not consider other types of deviations. However, there may be other deviations that could create problems for running an auction mechanism credibly. 

One such deviation is the possibility that the seller could approach one of the bidders and offer them a secret deviation to a different mechanism that is ex-ante beneficial to both parties. Despite this not being a safe deviation, we may nevertheless be concerned about the credibility of such an auction.\footnote{One defense against such unsafe deviations is that even if they are beneficial to the bidder in this auction, they can be detrimental to that bidder in the future or more generally undermine the credibility of that seller in the eyes of that buyer, leading to long-term losses. However, the analysis of credibility of auctions in repeated environments is beyond the scope of this paper.} Such deviations may be tempting to the seller even in the presence of public announcements: while public announcements can be verified, it may be impossible for bidders to verify that the seller has not contacted any other bidders ahead of time. 

Formally, an auctioneer-bidder joint deviation strategy $(\hat{S}_0, \hat{S}_j)$ is \textit{\textbf{safe}} if, for all agents $i \in N$, $i \neq j$, and for all type profiles $(\theta_0, \theta_N)$, there exists a pair $(\hat{\theta}_0, \hat{\theta}_{-i})$ such that 
\[o_i\big(\hat{S}_0, \hat{S}_j,  S_{-j}, \theta_0, \theta_N \big) = o_i\big(S_0, S_N, \hat{\theta}_0, (\theta_i, \hat{\theta}_{-i})\big)\,.\]

A joint deviation $(\hat{S}_0, \hat{S}_j)$ is  \textit{\textbf{mutually beneficial}} for auctioneer type $\theta^*_0$ and bidder $j$ if:
\begin{equation}
  \E_{\theta_N}\Big[u_0(S_0(\theta^*_0), S_N, \theta^*_0, \theta_N)\Big] < \E_{\theta_N}\Big[u_0(\hat{S}_0(\theta^*_0), \hat{S}_j, S_{-j}, \theta^*_0, \theta_N)\Big]\,,  \label{eq:strict}
\end{equation}
and for all $\theta_j$, playing according to $\hat{S}_j(\theta_j)$ is incentive compatible given $\big(\hat{S}_0(\theta^*_0), S_{-j}\big)$ and:
\begin{equation}
\sup_{\theta'_0} \Bigg\{\E_{\theta_{-j}}\Big[u_j(S_0(\theta'_0), S_N, \theta_N)\Big]\Bigg\} \leq \E_{\theta_{-j}}\Big[u_j(\hat{S}_0(\theta^*_0), \hat{S}_j, S_{-j},  \theta_N)\Big]\,. \label{eq:sup}
\end{equation}

We say that a protocol $(M, S_N)$ is \textit{\textbf{strongly credible}} if it is credible and for every bidder $j$ and every seller type $\theta_0$, there is no mutually beneficial safe joint deviation.

Admittedly, this is not the only possible definition of immunity of auction rules to secret side-deals between the auctioneer and a bidder. Our definition is motivated by an implicit assumption that the seller can approach at most one of the bidders and the secret deal is known to be offered before the auctioneer receives messages from any bidder (which is why in equations (\ref{eq:strict}) and (\ref{eq:sup}) the bidder and the auctioneer evaluate payoffs assuming that no other bidder changes their behavior and in case of a rejection, the approached bidder plays according to the original mechanism). We are focusing on such possible deviations because otherwise bidders would need to form beliefs about the messages the auctioneer received from others before that deviation, and their negative inferences about the messages of others could prevent them from accepting such side deals. Finally, note that our definition requires that the bidder finds the proposed deviation beneficial even if they assume that in case of a rejection of the side-deal the auctioneer's type will turn out to be the best for them (the sup on the left-hand side of (\ref{eq:sup})). 

We argue that the optimal English auction (with reserve prices that depend on the realized seller's cost) is strongly credible. 
However, the first-price auction may not be strongly credible even in the presence of public announcements. 

For the first claim, note that as we argued before, the optimal strategies of bidders do not depend on the seller's announcement about reserve prices (unless their value happens to be below the reserve price). Moreover, even if the seller learns something about the valuations of a subset of buyers, the optimal continuation mechanism is to revert to the English auction. See \Cref{app:ascending} for details. 

For the second claim (that even with public announcements, a first-price auction may not be strongly credible), return to the example from the introduction (\Cref{ex:fpa}). The seller with a public announcement would implement the optimal first-price auction in the following way: when they draw a cost of $0$, they would announce that the bidders can bid either $1$ or $\frac{5}{3}$. When they draw a cost of $0.7$, they would announce that only a bid of $2$ is allowed. If bidders believe that the seller will follow this mechanism, they would bid $b(1)=1$ and $b(2)=5/3$ when the reserve is low and $b(2)=2$ when the reserve is high. However, even with public announcements, this mechanism is not strongly credible: 

\begin{ex}[FPA with public announcements is not strongly credible] Consider the setting of \Cref{ex:fpa} and the public announcements as above. Consider the following deviation. Before the public announcement of the reserve, if the seller draws a low $\theta_0$, the seller would approach bidder $1$ and make them the following offer: You can bid either $1$ or $1.48$. If you bid $1$, I will announce publicly the reserve price of $1$. If you bid $1.48$, I will announce publicly a reserve price of $2$. If the other bidder beats you, you lose. However, if you bid $1.48$ and the other bidder does not meet the reserve price of $2$, I will secretly sell you the good at $1.48$ (in a private post-auction sale). This is clearly not a safe deviation, but we argue that it is mutually beneficial for buyer $1$ and the seller. 
First, suppose that the value of buyer $1$ is $2$. In the original mechanism, the buyer gets at most a profit of $(2-5/3) \times 3/4 = 1/4$ (they pay $5/3$ when they win and they win $3/4$ of the time in the case when the seller has a low $\theta_0$). In the new mechanism, they get an expected profit of \[(2-1.48) \times 0.5 > 1/4\,.\]
So that type prefers this secret deviation. Moreover, this joint deviation strategy is incentive compatible for the buyer: by misreporting that their value is $1$, the buyer can lower the reserve price to $1$, as in the original mechanism, but that yields the buyer a payoff of $1/4$, lower than the payoff from following the proposed strategy. 

Second, if the value of the buyer is $1$, we return to the original mechanism, and the maximum payoff the buyer can get is the same as what is offered. So for every type of buyer $1$, this is an improvement (and reporting the true value is incentive compatible in this case too). 

How about the seller? When buyer $1$ has a value $2$, the expected payoff in the original mechanism is $5/3$. In the new mechanism, it is $(1.48+2)/2 = 1.74 > 5/3$ (and when the value of buyer $1$ is $1$, the payoffs are unchanged).\footnote{If a reader is concerned that maybe this asymmetric mechanism is perhaps under commitment, better than the symmetric mechanism that we described before, note that the seller benefits only because they keep the deviation secret from bidder $2$: if bidder $2$ understood this deviation, then when the reserve they face is $1$, they would bid $1$ when their value is $2$, not $5/3$, and that would reduce the seller's expected revenue.} 
\hfill \qed 
\end{ex}

This example illustrates that even in the case that some announcements can be made publicly, there are still some deviations (albeit not fully safe) that could undermine the first-price auction.

\section{Conclusion}\label{sec:conclusion}

We study a seller with credibility concerns. We show that when the seller has private information about her cost, it is not possible to implement the optimal mechanism using a static mechanism. As we show, even the optimal first-price auction is no longer credible. We show that optimality requires a dynamic mechanism and that the English auction can be used to credibly implement the optimal mechanism. In contrast, we show that the Dutch auction may not be credible. For symmetric mechanisms in which only winners pay, we characterize all the static auctions that are credible: They are first-price auctions that depend only on the seller's cost ex post via a secret reserve, and may profitably pool bidders via a bid restriction. Our impossibility result highlights the role of public institutions and helps explain the use of dynamic mechanisms in informal auctions.

\newpage
\setlength\bibsep{12pt}
\bibliographystyle{ecta} 
\newpage
\bibliography{references}

\newpage
\appendix
\section{Omitted Proofs}

\subsection{Proof of \texorpdfstring{\Cref{lem:winner-paying}}{}}

We first claim the auctioneer's payoff for each cost type $\theta_0$ cannot exceed $1-\theta_0$. Note that credibility implies that the auctioneer of type $\theta_0$ has no incentive to mimic another type $\hat{\theta}_0$ when selecting which game $G$ to run. By the Envelope theorem, this implies that 
\[U(\theta_0) = \int^1_{\theta_0} Q(s) \d s + U(1)\,,\]
where $U(\,\cdot\,)$ is the equilibrium utility of the auctioneer and $Q(\,\cdot\,)$ is the probability of trade, which is pinned down by optimality. Therefore, the auctioneer's expected payoff is  
\[\E\big[U(\theta_0)\big] = \E\Big[\int^1_{\theta_0} Q(s) \d s\Big] + U(1)\,.\]
But under any optimal, first-price auction, the auctioneer's expected payoff is $\E\Big[\int^1_{\theta_0} Q(s) \d s\Big]$, and hence by optimality, we must have $U(1) = 0$. This immediately means that 
\[U(\theta_0) = \int^1_{\theta_0} Q(s) \d s \leq 1 - \theta_0\,.\]

Now, suppose for contradiction that there exists an optimal, credible, and static mechanism $(M, S_N)$ that is not winner-paying.  Then there exists some auctioneer type $\hat{\theta}_0$ and a positive-measure set of type profiles $\theta_N$ such that the type-$\hat{\theta}_0$ auctioneer gets strictly positive payment from some losing bidder $j \neq \tilde{y}(\theta_N, \hat{\theta}_0)$. Specifically, define $\mathcal{K} \subset \Theta_N$ as
\[\mathcal{K} := \Big\{\theta_N: \exists j \in N \text{ s.t. }  \tilde{y}(\theta_N, \hat{\theta}_0) \neq j \text{ and } \tilde{t}_j(\theta_N, \hat{\theta}_0) > 0 \Big\}\,,\]
which has positive measure. 
Let $\delta:= \E\Big[\1_{\theta_N \in \mathcal{K}} \sum_{j \neq  \tilde{y}(\theta_N, \hat{\theta}_0)} \tilde{t}_j(\theta_N, \hat{\theta}_0)\Big] > 0$ denote the expected payment the auctioneer of type $\hat{\theta}_0$ gets from the losing bidders.

Now, consider any auctioneer type $\theta_0 > 1 - \frac{1}{2}\delta$. By \textbf{Step 1(i)}, if playing by the book, then the auctioneer of type $\theta_0$ can get at most $1 - (1 - \frac{1}{2}\delta) = \frac{1}{2}\delta$. However, consider the following deviation by type $\theta_0$: 
\begin{itemize}
    \item Run the game $\hat{G} = M(\hat{\theta}_0)$.  
    \item If the outcome is such that no losing bidder pays the auctioneer, then tell every bidder $i$ that there exists some other bidder $j$ who played according to type $\hat{\theta}_j = 1$. 
    \item Otherwise, collect the payments from the losing bidders, and tell the winner that there exists some other bidder $j$ who played according to type $\hat{\theta}_j = 1$. 
\end{itemize}
Note that this deviation is safe. Moreover, by the optimality of $(M, S_N)$, when following this deviation, the auctioneer would keep the object almost surely and collect the positive payments from the losing bidders. Thus, the expected payoff from this deviation is at least $\delta > \frac{1}{2} \delta$. So $(M, S_N)$ cannot be credible. A contradiction.

\subsection{Proof of \texorpdfstring{\Cref{lem:ex-post}}{}}

Since $(M, S_N)$ is static, bidder $i$ has only one information set $I_i(\theta_0)$ that can depend on the seller's cost $\theta_0$. Let $\mathcal{I}^*_i:= \{I_i(\theta_0)\}_{\theta_0}$ be the possible information sets for bidder $i$ across seller's types $\theta_0$. We claim that for any bidder $i$, there exists $\tilde{b}_i(\theta_i, I_i)$ such that 
\[\tilde{t}_i(\theta_N, \theta_0) = b_i(\theta_i, \theta_0) = \tilde{b}_i(\theta_i, I_i(\theta_0))\]
for almost all types in $\Theta_N \times \Theta_0$ such that $\tilde{y}(\theta_N, \theta_0) = i$. Indeed, we can apply the argument in the proof of Theorem 1 in \citet{akbarpour2020credible} to the ``opponent type profile'' $(\tilde{\theta}_{-i}, \tilde{\theta}_0)$ for all $\tilde{\theta}_0$ in the same information set $I_i \in \mathcal{I}^*_i$. Intuitively, at any type profile $(\theta_N, \theta_0)$ such that $\tilde{y}(\theta_N, \theta_0) = i$, the seller has a safe deviation of charging bidder $i$ 
\[\sup_{\tilde{\theta}_{-i};\text{ } \tilde{\theta}_{0} \in I_i(\theta_0)} \tilde{t}_i(\theta_i, \tilde{\theta}_{-i}, \tilde{\theta}_0)\]
which cannot be detected by bidder $i$ (the measurability issue of such a deviation can be handled in the same way as in the proof of Theorem 1 in \citealt{akbarpour2020credible}). 

Now, let 
\[Q_i(\theta_i, I_i):= \E\big[Q_i(\theta_i, \theta_0)\mid I_i(\theta_0) = I_i \big]\,.\]
The BIC constraint for bidder $i$ implies that for all $I_i \in \mathcal{I}^*_i$, we have 
\[Q_i(\theta_i, I_i) \theta_i - Q_i(\theta_i, I_i) \tilde{b}_i(\theta_i, I_i)\geq Q_i(\hat{\theta}_i, I_i) \theta_i - Q_i(\hat{\theta}_i, I_i) \tilde{b}_i(\hat{\theta}_i, I_i)\,,\]
for all $\theta_i, \hat{\theta}_i$. This implies that $\tilde{b}_i(\theta_i, I_i)$ is non-decreasing in $\theta_i$ (by the Envelope theorem for any fixed $I_i$). Therefore, $b_i(\theta_i, \theta_0)$ is non-decreasing in $\theta_i$ for all $\theta_0$. Moreover, by the Envelope theorem, $b_i(\theta_i, \theta_0)$ must be continuous in $\theta_i$ at all $\theta_i, \theta_0$ where $Q_i(\theta_i, I_i(\theta_0)) > 0$ (given the Myersonian allocation from optimality). Moreover, when $Q_i(\theta_i, I_i(\theta_0)) = 0$, we can always let $b_i(\theta_i, \theta_0)$ be defined as $0$ (recall that by \Cref{lem:winner-paying}, the mechanism must be winner-paying). 

By credibility, it must be that if $\tilde{y}(\theta_N, \theta_0) = i$, then 
\[b_i(\theta_i, \theta_0) \geq \max\Big\{\theta_0, \max_{j\neq i} b_j(\theta_j, \theta_0)\Big\}\,,\]
because otherwise we have 
\[b_i(\theta_i, \theta_0) < \max\Big\{\theta_0, \max_{j\neq i} b_j(\theta_j, \theta_0)\Big\}\,,\]
and hence the seller has a profitable deviation by either walking away and claiming to every bidder that they are outbid by some other bidder (i.e., claiming that there exists a bidder with a higher type), or selecting bidder $j$ to get a strictly higher payment and claiming to bidder $i$ that he is outbid by some other bidder and claiming to bidder $j$ that he has the highest type.  

Now, we claim that for all $i, j$, $s\in[0, 1]$, and $\theta_0$, 
\[\max\big\{b_i(s, \theta_0), \theta_0\big\} = \max\big\{b_j(s, \theta_0), \theta_0\big\}\,.\]
Indeed, suppose for contradiction that this is not the case. Then, without loss of generality, we have 
\[ \max\big\{\theta_0, b_i(s, \theta_0)\big\} > \max\big\{\theta_0, b_j(s, \theta_0)\big\}\,,\]
and hence 
\[  b_i(s, \theta_0) > \max\big\{\theta_0, b_j(s, \theta_0)\big\}\,.\]
By construction of $b_i$, which equals $0$ at any $s < r(\theta_0)$, this implies that $s \geq r(\theta_0)$. Now take the type profile $\theta_N=(s,s,0,\dots,0)$, where the first two indices are $i$ and $j$ respectively, we have that under that type profile,
\[b_i(\theta_i, \theta_0) > \max\Big\{\theta_0, \max_{k\neq i} b_k(\theta_k, \theta_0)\Big\}\,.\]
Now by continuity, take the type profile $\theta'_N=(s, s+\varepsilon, 0,\dots,0)$ for some small $\varepsilon > 0$, where $\theta_j = s + \varepsilon$, we must also have 
\[b_i(\theta'_i, \theta_0) > \max\Big\{\theta_0, \max_{k\neq i} b_k(\theta'_k, \theta_0)\Big\}\,.\]
But at the type profile $\theta'_N$, the object must be allocated to bidder $j$ by optimality, which contradicts credibility by our previous observation. 

Moreover, note that for all $s < r(\theta_0)$, we have 
\[b_i(s, \theta_0) = 0\]
and for all $s \geq r(\theta_0)$, we must have
\[b_i(s, \theta_0) \geq \theta_0\,,\]
in order for the allocation to match the Myersonian allocation, given credibility. Now, for any $i, j$, and any $\theta_0$, we then have that for all $s < r(\theta_0)$, 
\[b_i(s, \theta_0) = b_j(s, \theta_0) = 0\]
and for all $s \geq r(\theta_0)$, 
\[b_i(s, \theta_0) = \max\big\{b_i(s, \theta_0), \theta_0\big\} =  \max\big\{b_j(s, \theta_0), \theta_0\big\} = b_j(s, \theta_0)\,.\]
Therefore, we have for all $i, j$, 
\[b_i(s, \theta_0) = b_j(s, \theta_0)\,.\]
Thus, there exists some function $b(\theta_i, \theta_0)$ such that the induced payment rule for any agent $i$ satisfies 
\[\tilde{t}_i(\theta_N, \theta_0) = b(\theta_i, \theta_0) \1_{\tilde{y}(\theta_N, \theta_0) = i}\,.\]
Thus the induced payment rule is symmetric across agents for every cost type $\theta_0$.

\subsection{Proof of \texorpdfstring{\Cref{thm:char}}{}}
Let $(M, S_N)$ be a symmetric, winner-paying, credible, and static mechanism. 

\textbf{Step 1.} We first make a sequence of observations about $(M, S_N)$. 

First, by \Cref{lem:psb}, $(M, S_N)$ must be a pay-as-bid auction. Without loss of generality, we also define $b_i(\theta_i, \theta_0) = 0$ for any $i$ and any $(\theta_i, \theta_0)$ such that $\P(\tilde{y}(\theta_i, \theta_{-i}, \theta_0, \varepsilon) = i) = 0$. 

Second, we claim that, under $(M, S_N)$, the bidding function $b_i(\theta_i, \theta_0)$ must be symmetric, i.e., there exists $b(\theta_i, \theta_0)$ such that $b_i(\theta_i, \theta_0) = b(\theta_i, \theta_0)$ for all $i$. By symmetry of $(M, S_N)$, the interim allocation probability
\[Q_i(\theta_i, \theta_0) = \E_{\theta_{-i}}\big[q_i(\theta_i, \theta_{-i}, \theta_0)\big]\]
must be symmetric across bidders, i.e., there exists some $Q(\theta_i, \theta_0)$ such that $Q_i(\theta_i, \theta_0) = Q(\theta_i, \theta_0)$ for all $i$.
But, since $(M, S_N)$ is pay-as-bid and winner-paying, and has symmetric interim payment rules across bidders conditional on $\theta_0$, this implies that for all $i$ and $j$
\[Q_i(s, \theta_0) b_i(s,  \theta_0)= Q_j(s, \theta_0) b_j(s, \theta_0) \]
almost everywhere. Moreover, recall that we set $b_i(\theta_i, \theta_0) = 0$ whenever $Q_i(\theta_i, \theta_0) = 0$. Hence, the bidding function $b_i(\theta_i, \theta_0)$ is symmetric across bidders almost everywhere.  

Third, we claim that, under $(M, S_N)$, any winner must have a maximal bid that exceeds the seller's private cost, i.e., for all $\theta_0$ and all $i$, if $\tilde{y}(\theta_N, \theta_0, \varepsilon) = i$, then 
\[b(\theta_i, \theta_0) \geq \max\Big\{\theta_0, \max_{j \neq i} b(\theta_j, \theta_0)\Big\}\,,\]
almost everywhere. Suppose for contradiction that this is not the case. Then there exists some $\theta_0$ and some bidder $i$ such that the set 
\[\mathcal{Q}:= \Big\{(\theta_N, \varepsilon): \tilde{y}(\theta_N, \theta_0, \varepsilon) = i, b(\theta_i, \theta_0) < \max\Big\{\theta_0, \max_{j \neq i} b(\theta_j, \theta_0)\Big\}\Big\}\]
has positive measure. But consider the following deviation by the auctioneer of type $\theta_0$: 
\begin{itemize}
    \item Run the game $M(\theta_0)$.
    \item If $\tilde{y}(\theta_N, \theta_0, \varepsilon) = i$ and $b(\theta_i, \theta_0) < \theta_0$, keep the object. 
    \item Otherwise, if $\tilde{y}(\theta_N, \theta_0, \varepsilon) = i$ and $b(\theta_i, \theta_0) <  \max_{j \neq i}  b(\theta_j, \theta_0)$, allocate the object to bidder $j$ with the highest bid $b(\theta_j, \theta_0)$, instead of bidder $i$, and charge  bidder $j$ a payment $b(\theta_j, \theta_0)$.
\end{itemize}
This is clearly a profitable deviation. We argue that this is also safe. By the symmetry of the induced allocation rule, for any bidder $i$ and any $(\theta_i, \theta_0)$, there exist a type profile $\theta'_{-i}$ and realization $\varepsilon'$ such that bidder $i$ loses and pays zero. Moreover, for any bidder $j$ and any $(\theta_j, \theta_0)$ such that $b(\theta_j, \theta_0) > b(\theta_i, \theta_0) \geq 0$, we have $\P(\tilde{y}(\theta_{j}, \theta_{-j}, \theta_0, \varepsilon) = j) > 0$ by the construction of $b$. Hence, there exists some type profile $\theta'_{-j}$ and some realization $\varepsilon'$ such that bidder $j$ wins and pays $b(\theta_j, \theta_0)$. Thus, the deviation is safe. But then $(M, S_N)$ cannot be credible. A contradiction. 

Fourth, we claim that, under $(M, S_N)$, if the maximal bid exceeds the seller's cost, then the seller must allocate the object to some bidder, i.e., if $\max_i b(\theta_i, \theta_0) > \theta_0$, then \[\tilde{y}(\theta_N, \theta_0, \varepsilon) \neq 0\,,\]
almost everywhere. Suppose for contradiction that this is not the case. Then there exists some $\theta_0$ such that 
\[\mathcal{Q}':= \Big\{(\theta_N, \varepsilon): \tilde{y}(\theta_N, \theta_0, \varepsilon) = 0, \max_i b(\theta_i, \theta_0) > \theta_0 \Big\}\]
has positive measure. But consider the following deviation by the auctioneer of type $\theta_0$: 
\begin{itemize}
    \item Run the game $M(\theta_0)$.
    \item If $\tilde{y}(\theta_N, \theta_0, \varepsilon) = 0$ and $\max_i b(\theta_i, \theta_0) > \theta_0 $, allocate the object to bidder $i$ with the highest bid $b(\theta_i, \theta_0)$, and charge bidder $i$ a payment $b(\theta_i, \theta_0)$. 
\end{itemize}
This is clearly a profitable deviation. It is also safe by the same argument in the previous observation. But then $(M, S_N)$ cannot be credible. A contradiction.

\textbf{Step 2.} Let
\[B := \Big\{b(\theta_i, 0): \theta_i \in \Theta_i, \, Q(\theta_i, 0) > 0\Big\}\,.\]
We show that $(M, S_N)$ must be outcome-equivalent to a first-price auction with a walk-away option and the public bid space $B$. 

The proof of this claim involves three substeps. 

\textbf{Step 2(i).} First, for each $\theta_0$, let $G(b;\theta_0)$ denote the CDF of the random variable $b(\theta_i, \theta_0)$. We claim that 
\[G(s; \theta_0) = G(s; 0)\]
for all $s \in [\theta_0, 1]$. To prove it, define 
\[\Phi(s; \theta_0) = \int_{s}^1 (b - s) \d G(b; \theta_0)\,.\]
Fix any $\hat{\theta}_0 \in [0, 1]$. Let $G_0(\,\cdot\,) := G(\,\cdot\,; 0)$ and $\hat{G}(\,\cdot\,) := G(\,\cdot\,; \hat{\theta}_0)$. We first show that 
\[\int_s^1  (b - s) \d G_0(b) \geq \int_s^1  (b - s) \d \hat{G}(b) \]
for a $G_0$-measure-one set of $s$. Suppose for contradiction that this is not the case. Then there exists a $G_0$-positive-measure set $S\ni s$ such that 
\[\int_s^1  (b - s) \d G_0(b) < \int_s^1  (b - s) \d \hat{G}(b)\,.\]
Now, consider the following deviation by the auctioneer of type $0$: 
\begin{itemize}
    \item  Run the game $M(0)$.
    \item  If the maximal bid of the first $|N| - 1$ bidder, $\max_{j < |N|} b(\theta_j, 0)$, is in the set $S$, then give the information set $\mathcal{I}_{|N|}(\hat{\theta}_0)$ to the last bidder. 
\end{itemize}
Since $S$ is a $G_0$-positive-measure set, by symmetry and independence, the above event is a $G_0$-positive-measure set. For the last bidder, upon receiving the information set $\mathcal{I}_{|N|}(\hat{\theta}_0)$, his payment-conditional-on-winning is $b(\theta_{|N|}, \hat{\theta}_0)$. For any maximal bid by the first $|N|-1$ bidders $s$, if playing by the book, by \textbf{Step 1}, we know that the auctioneer of type $0$ gets 
\[s + \int (b - s) \1_{b \geq s} \d G_0(b)\,.\]
On the other hand, if following this deviation, by \textbf{Step 1}, the auctioneer gets 
\[s + \int (b - s) \1_{b \geq s} \d \hat{G}(b)\,,\]
which is strictly higher whenever $s \in \mathcal{S}$ (which happens with a positive probability). Therefore, type-$0$ auctioneer has a profitable safe deviation, contradicting the credibility of $(M, S_N)$. 

Now, note that 
\[\Phi(s; \theta_0) = \int_s^1  (b - s) \d G(b; \theta_0) = \int_s^1  (1 - G(b; \theta_0)) \d b\]
is a non-increasing and convex function in $s$, for any $\theta_0$. By continuity, we have that 
\[\Phi(s; 0) \geq \Phi(s; \hat{\theta}_0)\]
for all $s \in \supp(G_0)$. We claim that we also have the same inequality for all $s \geq \min \{\supp(G_0)\}$. Suppose for contradiction that there exists some $s \geq \min \{ \supp(G_0) \}$ such that 
\[\Phi(s; 0) < \Phi(s; \hat{\theta}_0)\,.\]
Then $s\not\in \supp(G_0)$, and hence $s$ must be in an open interval $(s_1, s_2) \subset ([0, 1] \backslash \supp(G_0))$ such that $s_1 \in \supp(G_0)$ and $s_2\in \supp(G_0)$. In particular, $G_0$ is constant on the open interval, and hence 
\[\Phi(s; 0)  =  \frac{s_2 - s}{s_2 - s_1} \Phi(s_1; 0) +  \frac{s - s_1}{s_2 - s_1}\Phi(s_2; 0)\,.\]
We also know that 
\[\Phi(s_1; 0) \geq \Phi(s_1; \hat{\theta}_0)\,, \qquad  \Phi(s_2; 0) \geq \Phi(s_2; \hat{\theta}_0)\,. \]
But then 
\[\Phi(s; \hat{\theta}_0) > \Phi(s; 0) =  \frac{s_2 - s}{s_2 - s_1} \Phi(s_1; 0) +  \frac{s - s_1}{s_2 - s_1}\Phi(s_2; 0) \geq \frac{s_2 - s}{s_2 - s_1} \Phi(s_1; \hat{\theta}_0) +  \frac{s - s_1}{s_2 - s_1}\Phi(s_2; \hat{\theta}_0)\]
contradicting the convexity of $\Phi(\,\cdot\,; \hat{\theta}_0)$. 

Similarly, we also claim that for all $s \geq \max\{\min\{\supp(\hat{G})\}, \hat{\theta}_0\}$, 
\[\Phi(s; \hat{\theta}_0) \geq \Phi(s; 0)\,.\]
The proof is exactly symmetric to the above argument if $\min\{\supp(\hat{G})\} \geq \hat{\theta}_0$. Now, suppose that $\min\{\supp(\hat{G})\} < \hat{\theta}_0$. Then, there exists a $\hat{G}$-positive-measure event under which the maximal bid from the first $|N|-1$ bidders is strictly less than $\hat{\theta}_0$. If that happens, the auctioneer of type $\hat{\theta}_0$ can deviate to give the last bidder information set $\mathcal{I}_{|N|}(0)$. In order for this deviation not to be profitable, we must have 
\[\Phi(\hat{\theta}_0; \hat{\theta}_0) \geq \Phi(\hat{\theta}_0; 0)\,.\]
For all $s > \hat{\theta}_0$ such that $s \in \supp(\hat{G})$, the same deviation as before yields the desired inequality. For all $s > \hat{\theta}_0$ such that $s \not\in \supp(\hat{G})$, the same convexity argument as above would also yield the desired inequality (since we also know that the inequality holds at $\hat{\theta}_0$). Therefore, we have 
\[\Phi(s; \hat{\theta}_0) \geq \Phi(s; 0)\]
for all $s \geq \max\{\min\{\supp(\hat{G})\}, \hat{\theta}_0\}$. 

Combining these two sets of inequalities together, we have 
\[\Phi(s; \hat{\theta}_0) = \Phi(s; 0) \]
for all 
\[s \geq \max\big\{\min\{\supp(G_0)\}, \min\{\supp(\hat{G})\}, \hat{\theta}_0 \big\}\,.\]
Denote $m_0 = \min\{\supp(G_0)\}$ and $\hat{m} = \min\{\supp(\hat{G})\}$. Consider first the case $\hat{\theta}_0 < \max\{m_0, \hat{m}\}$. Then, we have  
\begin{equation}
    \Phi(s; \hat{\theta}_0) = \Phi(s; 0)  \label{eq:equal}
\end{equation}
for all $s \geq \max\{m_0, \hat{m}\}$. We claim that this implies  
\[\Phi(s; \hat{\theta}_0) = \Phi(s; 0)\]
for all $s \in [\hat{\theta}_0, 1]$. We consider two subcases. First, suppose that $m_0 \geq \hat{m}$. Then, by \eqref{eq:equal}, $G_0$ must first-order stochastically dominate $\hat{G}$. For any $s \geq \max\{\hat{\theta}_0, \hat{m}\}$, we have 
\[\Phi(s; \hat{\theta}_0) \geq \Phi(s; 0)\,,\]
but the FOSD implies that 
\[\Phi(s; \hat{\theta}_0) \leq \Phi(s; 0)\,.\]
Thus, $\Phi(s; \hat{\theta}_0) = \Phi(s; 0)$. Now, for any $s \in [ \hat{\theta}_0, \max\{\hat{\theta}_0, \hat{m}\} )$, by construction, 
\[G_0(s) = \hat{G}(s) = 0\,.\]
It follows immediately that $\Phi(s; \hat{\theta}_0) = \Phi(s; 0)$ for all $s \in [\hat{\theta}_0, 1]$. Now suppose that $\hat{m} > m_0$. Then, by \eqref{eq:equal}, $\hat{G}$ must first-order stochastically dominate $G_0$. The symmetric argument of the previous case implies that $\Phi(s; \hat{\theta}_0) = \Phi(s; 0)$ for all $s \in [\hat{\theta}_0, 1]$. As a consequence, we immediately have that 
\[G(s; \hat{\theta}_0) = G(s; 0)\]
for all $s \in [\hat{\theta}_0, 1]$ by taking the derivative of $\Phi$ with respect to $s$ and using that $G$ is right-continuous in $s$. 

Now, consider the case  $\hat{\theta}_0 \geq \max\{m_0, \hat{m}\}$. The same argument immediately implies that 
\[G(s; \hat{\theta}_0) = G(s; 0)\]
for all $s \in [\hat{\theta}_0, 1]$. 

Since $\hat{\theta}_0$ is arbitrary, this proves the claim. Before we move on, we make one more observation: We claim that for all $\hat{\theta}_0 > 0$, we have  
\[ G(\hat{\theta}^{-}_0; \hat{\theta}_0) \geq G(\hat{\theta}^{-}_0; 0)\,,\]
where $\hat{\theta}^{-}_0$ denotes the left limit of $\hat{\theta}_0$. Indeed, suppose not. Then, there exists an open interval $(s', \hat{\theta}_0)$ such that for all $s \in (s', \hat{\theta}_0)$ we have 
\[ G(s; \hat{\theta}_0) < G(s; 0)\,,\]
which then, combined with our previous claim, implies that 
\[\int^1_s (1 - G(b; \hat{\theta}_0))\d b >\int^1_s (1 - G(b; 0))\d b \,.\]
However, since $G(s; 0) > 0$, we have $s \geq \min\{\supp(G_0)\}$, and hence by our previous argument, 
\[\int^1_s (1 - G(b; \hat{\theta}_0))\d b \leq \int^1_s (1 - G(b; 0))\d b \,,\]
a contradiction.

\textbf{Step 2(ii).} By BIC and credibility, as shown in the proof of \Cref{lem:ex-post}, $b(\theta_i, \theta_0)$ is non-decreasing in $\theta_i$ for all $\theta_0$. By \textbf{Step 2(i)}, for any $\theta_0$, we have
\[\max\big\{b(\theta_i, 0), \theta_0\big\} \eqid \max\big\{b(\theta_i, \theta_0), \theta_0\big\}\,.\]
Since both $\max\big\{b(\theta_i, 0), \theta_0\big\}$ and $\max\big\{b(\theta_i, \theta_0), \theta_0\big\}$ are non-decreasing in $\theta_i$, the above implies that for any $\theta_0$, we have 
\[\max\big\{b(\theta_i, 0), \theta_0\big\} = \max\big\{b(\theta_i, \theta_0), \theta_0\big\}\]
almost everywhere in $\Theta_i$. Moreover, by \textbf{Step 2(i)}, for all $\theta_0 > 0$, we have  
\[G(\theta^{-}_0; \theta_0) \geq G(\theta^{-}_0; 0)\,,\]
which implies that 
\[\P_{\theta_i}\big(b(\theta_i, 0) = \theta_0\big) \geq \P_{\theta_i}\big(b(\theta_i, \theta_0) = \theta_0\big)\,.\]
Since $b(\,\cdot\,, 0)$ and $b(\,\cdot\,, \theta_0)$ are both non-decreasing, this then implies that $b(\theta_i, 0) = \theta_0$ for almost all $\theta_i$ with $b(\theta_i, \theta_0) = \theta_0$. 

Now, we claim that the auctioneer can replicate the outcomes by using another protocol $(M', S'_N)$ that is a first-price auction with a walk-away option and a public bid space $B := \Big\{b(\theta_i, 0): \theta_i \in \Theta_i, \, Q(\theta_i, 0) > 0\Big\}$ (with the strategy profile given by $\big\{b(\theta_i, 0) \big\}_{i \in N}$).  

By the previous argument, there exists a measure-$1$ set of type profiles on which we have for all $i \in N$
\begin{equation}
    \max\big\{b(\theta_i, 0), \theta_0\big\} = \max\big\{b(\theta_i, \theta_0), \theta_0\big\}\,. \label{eq:key}
\end{equation}
Fix this set of type profiles. By \textbf{Step 1}, if
\[\max_i\big\{b(\theta_i, \theta_0)\big\} > \theta_0\]
then the object must be allocated to a maximal bidder for whom we have $b(\theta_i, \theta_0) = b(\theta_i, 0)$ by \eqref{eq:key}. Also by \eqref{eq:key}, the rest of the bidders either bid at some $b(\theta_j, \theta_0) = b(\theta_j, 0)$ or bid weakly below $\theta_0$. In the FPA that we are constructing, let the auctioneer  follow the same tie-breaking rule as in $(M, S_N)$ (such tie-breaking may require cheap-talk messages from the bidders to report their types). It follows that, at this type profile, our construction would yield the same ex-post allocation and the same ex-post payment. 

Similarly, by \textbf{Step 1}, if
\[\max_i\big\{b(\theta_i, \theta_0)\big\} < \theta_0\]
then the object must be kept by the auctioneer in $(M, S_N)$. At any such type profile, we must have that $b(\theta_i, 0) \leq \theta_0$ for all bidders $i$ by \eqref{eq:key}. Thus, in the FPA that we are constructing, let the auctioneer always keep the object at such a type profile, resulting in the same ex-post allocation and payment. 

Finally, if
\[\max_i\big\{b(\theta_i, \theta_0)\big\} = \theta_0\,,\]
then by \textbf{Step 2(i)} and the previous observation, it must be that 
\[\max_i\big\{b(\theta_i, 0)\big\} = \theta_0\,.\]
Thus, in the FPA that we are constructing, we may let the auctioneer follow the same ex-post allocation and payment rule as in the $(M, S_N)$. 

Now, we claim that the strategy profile $\big\{b(\theta_i, 0) \big\}_{i \in N}$ continues to form an equilibrium. Fix any bidder $i$. Note that if the auctioneer can send a cheap-talk message $I_i(\theta_0)$ to each bidder $i$ while providing the bid space $B$, then bidder $i$ would have exactly the same information in $(M', S'_N)$ as in $(M, S_N)$. In such a case, for each cheap-talk message, fixing the opponent's strategies $\big\{b(\theta_{j}, 0)\big\}_{j \neq i}$, bidder $i$'s any feasible strategy of misreporting $\theta_i$ would result in the same ex-post outcome as in $(M, S_N)$, almost everywhere. Moreover, bidder $i$ has the same belief about $(\theta_0, \theta_{-i})$, and hence following the strategy $b(\theta_i, 0)$ by truthfully reporting must be a best response for type $\theta_i$. But then, since this strategy does not depend on the cheap-talk message $I_i(\theta_0)$, it must also maximize bidder $i$'s expected payoff even if bidder $i$ does not observe the cheap-talk message $I_i(\theta_0)$. Hence, $(M', S'_N)$ would also be BIC. 

\textbf{Step 2(iii).} Finally, we complete the characterization by noting that the event 
\[\max_{i\in N} b(\theta_i, 0) = \theta_0\]
can only happen with zero probability, given the independence of $\theta_0$ and $\theta_N$. Therefore, in $(M', S'_N)$, we may let the auctioneer walk away if and only if 
\[\max_{i\in N} b(\theta_i, 0) \leq \theta_0\,,\]
while keeping the resulting outcomes to be equivalent almost everywhere. 

\section{Ascending Auctions}\label{app:ascending}

In this appendix, for completeness, we give the formal definition of an ascending auction following \citet{akbarpour2020credible}. We then show that optimal ascending auctions are strongly credible.

\subsection{Definition of Ascending Auctions}

As in \citet{akbarpour2020credible}, we assume that the type space is discrete to avoid modeling continuous-time games. Let $\Theta_i := \{\theta^1_i,\dots, \theta^K_i\}$, where $\theta^1_i = 0$ and $\theta^{k+1}_i - \theta^k_i > 0$, for all $k$ and all $i = 0, 1, \dots, N$ (including the auctioneer's types).

We say that $(M, S_N)$ is an \textit{\textbf{ascending auction}} if for every $\theta_0$, $\big( M(\theta_0), S_N\mid_{\mathcal{I}_N(\theta_0)} \big)$ satisfies Definition 14 in \citet{akbarpour2020credible}.

\subsection{Strong Credibility of Ascending Auctions}
In this section, we provide a generalization of Lemma 3 in \citet{akbarpour2020credible}, which will be used to prove the strong credibility of optimal ascending auctions. 

\begin{lemma}\label{lem:bestresponse}
Let $(M, S_N)$ be an ascending auction. For every bidder $i$, if $(\hat{S}_0, \hat{S}_j)$ is a safe joint deviation such that $j \neq i$, then $S_i$ is an ex post best response to $(\hat{S}_0, \hat{S}_j, S_{- \{ j,i\}})$ for all $\theta_0$ and $\theta_{-i}$. 
\end{lemma}
\begin{proof}[Proof of \Cref{lem:bestresponse}]
Let $(\hat{S}_0, \hat{S}_j)$ be a safe joint deviation. Take any type $\theta_i$. We claim that any deviating strategy $\hat{S}_i(\theta_i)$ cannot yield strictly higher payoff for type $\theta_i$. 

Suppose that $S_i(\theta_i)$ and deviating
strategy $\hat{S}_i(\theta_i)$ choose different actions for the first time after receiving message $I_i$. There
are three cases to consider; we will show that, in each case, $\hat{S}_i(\theta_i)$ is not a profitable
deviation for every possible realization of $\theta_0$ and $\theta_{-i}$. 
\newline
\textbf{Case 1:} Suppose that at $I_i$, strategy $\hat{S}_i$ chooses the quit action, thus receiving zero utility.  Since $(\hat{S}_0, \hat{S}_j)$ is safe and the ascending auction has threshold pricing, the strategy $S_i$ must obtain weakly positive utility by accepting, so the deviation is unprofitable. 
\newline
\textbf{Case 2:} Suppose that at $I_i$, strategy $S_i$ quits while deviation $\hat{S}_i$ accepts. Following the same logic as in the proof of Lemma 3 of \citet{akbarpour2020credible}, we can find plausible explanations $\hat{\theta}_0, \hat{\theta}_{-i}$ and an agent type $\hat{\theta}_i$ such that 
\[o_i(\hat{S}_0, \hat{S}_j, \hat{S}_i, S_{-\{j,i\}}, \theta_0, \theta_N) = o_i(S_0, S_N, \hat{\theta}_0, (\hat{\theta}_i, \hat{\theta}_{-i}))\,.\]
If $\tilde{y}(\hat{\theta}_i, \hat{\theta}_0, \hat{\theta}_{-i}) \neq i$, then the deviation is clearly unprofitable by threshold pricing. Now suppose $\tilde{y}(\hat{\theta}_i, \hat{\theta}_0, \hat{\theta}_{-i}) = i$. However, note that $\tilde{y}(\theta_i, \hat{\theta}_0, \hat{\theta}_{-i}) \neq i$ since $S_i(\theta_i)$ can also reach $I_i$ and specifies agent $i$ to quit. But then the misreport $\hat{\theta}_i$ induces the agent to win the object and pay $\tilde{t}_i(\hat{\theta}_i, \hat{\theta}_0, \hat{\theta}_{-i}) > \theta_i$ by threshold pricing and the orderly property of the ascending auction, so the deviation is unprofitable. 
\newline
\textbf{Case 3:} Suppose that at $I_i$, the two strategies decide to accept under two different actions. Then, property 5b of the ascending auction guarantees that the two strategies generate the same utility.
\end{proof}

The following result strengthens \Cref{thm:english} to the notion of strong credibility. 

\begin{theorem}\label{thm:strongenglish}
There exists an optimal, strongly credible auction. In particular, an optimal, ascending auction is strongly credible. 
\end{theorem}

\begin{proof}[Proof of \Cref{thm:strongenglish}]
Suppose that $(M, S_N)$ is an optimal, ascending auction. By \Cref{thm:english}, we know that $(M, S_N)$ is credible. Now, to prove strong credibility, suppose for contradiction that there exists a safe joint deviation  $(S', \hat{S}_j)$ that is mutually beneficial between the auctioneer of type $\theta_0^*$ and bidder $j$.

Consider the following strategy by the auctioneer. If the auctioneer's type is $\theta_0 \neq \theta_0^*$, run the optimal ascending auction as in $S_0$. If the auctioneer's type is $\theta_0^*$, deviate to $S'(\theta_0^*)$ and ask bidder $j$ to play according to $\hat{S}_j$. By the definition of mutually beneficial deviation, $\hat{S}_j$ must be optimal for bidder $j$ to play against $(S'(\theta_0^*), S_{-j})$. Moreover, by \Cref{lem:bestresponse}, we also know that $S_i$ is a best reply to  $(S', \hat{S}_j, S_{-\{j, i\}})$ even after knowing the auctioneer's type. Therefore, we know that 
\[\big(S'(\theta_0^*),  \hat{S}_j,  S_{-j} \big)\]
induces a BIC mechanism (where everyone knows the auctioneer's type is $\theta_0^*$) that yields the auctioneer a payoff strictly higher than what the auctioneer can obtain following $S_0(\theta^*_0)$. But then if the auctioneer follows the above strategy, then the auctioneer can guarantee an expected payoff in a BIC mechanism (where the auctioneer's type is publicly observable) that is strictly higher than what she could obtain in the optimal, ascending auction (recall that we have discrete types), contradicting that the ascending auction is optimal. \end{proof}
\end{document}